\documentclass[onecolumn]{IEEEtran}
\usepackage[pdftex]{graphicx}
\usepackage{algorithmic}
\usepackage{amsmath}
\usepackage{amssymb}
\usepackage{amsthm}
\usepackage{multirow}
\usepackage{color}
\usepackage{longtable}
\usepackage{array}
\usepackage{url}
\usepackage[ruled,linesnumbered]{algorithm2e}
\usepackage{enumerate}
\usepackage{float}
\usepackage{mathtools}
\usepackage{bbm}
\usepackage{hyperref}
\usepackage{cite}
\usepackage{comment}
\newtheorem{theorem}{Theorem}[section]
\newtheorem{definition}[theorem]{Definition}
\newtheorem{lemma}[theorem]{Lemma}
\newtheorem{example}[theorem]{Example}

\newtheorem{corollary}[theorem]{Corollary}
\newtheorem{remark}[theorem]{Remark}
\newtheorem{observation}[theorem]{Observation}
\newtheorem{problem}[theorem]{Problem}
\newtheorem{claim}[theorem]{Claim}

\begin{document}
\title{Sequence Reconstruction under Channels with Multiple Bursts of Insertions or Deletions}

\author{Zhaojun Lan, Yubo~Sun, Wenjun~Yu, and Gennian~Ge%
\thanks{The research of G. Ge was supported by the National Key Research and Development Program of China under Grant 2020YFA0712100, the National Natural Science Foundation of China under Grant 12231014, and Beijing Scholars Program.}
\thanks{Z. Lan ({\tt 2200501014@cnu.edu.cn}), Y. Sun ({\tt 2200502135@cnu.edu.cn}), and G. Ge ({\tt gnge@zju.edu.cn}) are with the School of Mathematical Sciences, Capital Normal University, Beijing 100048, China.}
\thanks{W. Yu ({\tt wenjun@post.bgu.ac.il}) is with the School
   of Electrical and Computer Engineering, Ben-Gurion University of the Negev, Beer Sheva 8410501, Israel.}
}

\date{}\maketitle

\begin{abstract}
The sequence reconstruction problem involves a model where a sequence is transmitted over several identical channels. This model investigates the minimum number of channels required for the unique reconstruction of the transmitted sequence. 
Levenshtein established that this number exceeds the maximum size of the intersection between the error balls of any two distinct transmitted sequences by one.
In this paper, we consider channels subject to multiple bursts of insertions and multiple bursts of deletions, respectively, where each burst has an exact length of value $b$.
We provide a complete solution for the insertion case while partially addressing the deletion case.
Specifically, our key findings are as follows:
\begin{itemize}
  \item \textbf{Insertion Case:} We investigate $b$-burst-insertion balls of radius $t$ centered at $q$-ary sequences of length $n$. 
  Firstly, we demonstrate that the size of the error ball is independent of the chosen center and is given by:
  \begin{align*}
    I_{q,b}(n,t)= q^{t(b-1)}\sum_{i=0}^t\binom{n+t}{i}(q-1)^i.
  \end{align*}
  Next, we calculate the maximum intersection size between two distinct error balls, determining it to be: \[
  N_{q,b}^+(n,t)= q^{t(b-1)}\sum_{i=0}^{t-1}\binom{n+t}{i}(q-1)^i[1-(-1)^{t-i}].
  \]
  Furthermore, we show that the intersection size between error balls centered at two sequences that differ only at the first position can reach this maximum value. Finally, we propose a reconstruction algorithm with a linear runtime complexity of $O_{q,t,b}(nN)$ which processes $N\geq N_{q,b}^{+}(n,t)+1$ distinct output sequences from the channel to recover the correct transmitted sequence.

  \smallskip
  
  \item \textbf{Deletion Case:} We examine $b$-burst-deletion balls of radius $t$ centered at $q$-ary sequences of length $n$. Unlike burst-insertion ball, the size of a $b$-burst-deletion ball depends on the chosen center.
      To this end, we first assess which error ball has a larger size and demonstrate that the error ball centered at the $b$-cyclic sequence $0^b \circ 1^b \circ \cdots \circ (q-1)^b \circ 0^b\cdots$ has the largest size. This size is expressed as:
  \[D_{q,b}(n,t)= \sum_{i=0}^t\binom{n-bt}{i}D_{q-1,1}(t,t-i),
  \]
  where $D_{1,1}(t,t-i)=1$ for $i\in [0,t]$.
  We then investigate the maximum intersection size between two distinct error balls, specifically addressing the case where $q=2$, $b\geq 2$, and $n\geq b(t+1)-1$.
  We derive this value as: 
  \[N_{2,b}^{-}(n,t)= D_{2,b}(n,t)- D_{2,b}(n-b,t)+ D_{2,b}(n-3b,t-2).
  \]
  Moreover, we show that the intersection size between error balls of $0^{b} \circ 1^b \circ 0^b \circ 1^b\cdots $ and $0^{b-1} \circ 1 \circ 1^b \circ 0^b \circ 1^b\cdots $ equals $N_{2,b}^{-}(n,t)$.
  Finally, we propose a reconstruction algorithm with a linear runtime complexity of $O_{t,b}(nN)$ that processes $N\geq N_{2,b}^{-}(n,t)+1$, distinct output sequences from the channel to return the correct transmitted sequence.
\end{itemize}
\end{abstract}


\section{Introduction}\label{Introduction}
The sequence reconstruction problem corresponds to a model where a sequence from a specific codebook is transmitted over several identical noisy channels, with each channel producing a distinct output. The outputs received by the decoder are then used to reconstruct the transmitted sequence. The primary challenge within this framework is to determine the minimum number of channels required to ensure the unique reconstruction of the transmitted sequence.
This problem can be framed as identifying the size of the largest intersection among the output sets of the channels after transmitting two distinct sequences. 
The sequence reconstruction problem was first introduced by Levenshtein \cite{Levenshtein-01-JCTA-recons,Levenshtein-01-IT-recons} in 2001 and has gained renewed interest due to its applications in DNA-based data storage \cite{Church-12-science-DNA, Goldman-13-nature-DNA, Yazdi-15-TMBMC-DNA, Organick-18-nature-DNA} and racetrack memories \cite{Parkin-08-since-RM, Chee-18-IT-RM}.

When the codebook encompasses the entire space, specifically the set of all sequences of length $n$, 
Levenshtein's seminal work \cite{Levenshtein-01-JCTA-recons} addressed the sequence reconstruction problem concerning channels that introduce multiple insertions and deletions, respectively. 
Following this, Abu-Sini \textit{et al.} \cite{Sini-21-IT-reconstr} examined the problem with channels that introduce one insertion and one substitution, while Sun \textit{et al.} \cite{Sun-23-IT-BDR} focused on channels that introduce a burst of insertions and a burst of deletions, respectively.

In scenarios where the codebook is a subset of the full space, particularly for certain error-correcting codes, Sala \textit{et al.} \cite{Sala-17-IT-reconstr-ins} tackled the sequence reconstruction problem under channels that introduce multiple insertions, with the codebook being any $t$-insertion correcting code for $t \geq 1$.
Ye \textit{et al.} \cite{Ye-23-IT} also considered channels that introduce multiple insertions but focused on codebooks satisfying the property that the intersection size between two distinct single insertion balls is at most one.
Gabrys and Yaakobi \cite{Gabrys-18-IT-reconstr-del} addressed channels that introduce multiple deletions, using a single-deletion correcting code as the codebook. 
Additionally, Chrisnata \textit{et al.} \cite{Chrisnata-22-IT-reconstr} investigated channels that introduce two deletions, requiring codebooks to satisfy the property that the intersection size between any two distinct single deletion balls is at most one.
Recently, Pham \textit{et al.} \cite{Pham-25-JCTA-reconstr-del} generalized the work of Gabrys and Yaakobi by investigating channels that introduce multiple deletions, with the codebook being any $t$-deletion correcting code for $t\geq 1$, and provided a complete asymptotic solution.

In addition to the work related to sequence reconstruction problem under insertion  and deletion channels mentioned above, researchers have also considered channels that introduce other types of errors, such as substitutions \cite{Levenshtein-01-IT-recons}, limited magnitude errors \cite{Wei-22-IT-reconstr}, and tandem duplications \cite{Yehezkeally-20-IT-recons-tandem-dupli}, among others.
Furthermore, as a relaxation of the decoding requirement, the list decoding problem under the reconstruction model has also been explored in \cite{Yaakobi-18-IT-reconstr, Junnila-14-IT-reconstr, Junnila-16-IT-reconstr, Junnila-21-IT-reconstr, Junnila-24-IT-reconstr, Yehezkeally-20-IT-recons-tandem-dupli, Sini-24-IT-reconstr}.
Another variant is discussed in references \cite{Chee-18-IT-RM, Cai-22-IT-recon-edit, Sun-23-IT-reconstr, Sun-24-arXiv, Chrisnata-22-IT-reconstr, Ye-23-IT, Wu-24-DCC-reconstr, Zhang-24-ISIT-reconstr}, which investigate a dual problem of the sequence reconstruction framework. 
They fixed the number of channels and designed codebooks that enable the receiver to uniquely reconstruct the transmitted sequence.

Bursts of insertions and deletions are common classes of errors observed in various applications, such as DNA storage and racetrack memories. The design of codebooks capable of correcting bursts of insertions or deletions of length exactly $b$ is explored in the references \cite{Sima-20-ISIT, Cheng-14-ISIT-BD, Schoeny-17-IT-BD, Schoeny-17, Saeki-18-ISIT-BD, Nguyen-24-IT-BD, Sun-24-arXiv-BD, Sun-25-IT, Ye-24-arXiv-BD}.
In this paper, we investigate the sequence reconstruction problem under channels that introduce $t$ bursts of insertions and $t$ bursts of deletions, respectively, where each burst is of length exactly $b$, with the codebook being the set of all sequences of length $n$.
We note that our work generalizes the contributions of Levenshtein \cite{Levenshtein-01-JCTA-recons} (for the case where $b=1$) and Sun \textit{et al.} \cite{Sun-23-IT-BDR} (for the case where $t=1$).

\subsection{Bursts of Insertions}
Let $\Sigma_q^n$ denote the set of all sequences of length $n$ over $\Sigma_q$.
For a given sequence $\boldsymbol{x}=x_1\cdots x_n\in \Sigma_q^n$ and a positive integer $b\geq 1$, we say that $\boldsymbol{x}$ suffers \textbf{a burst of $b$ insertions} (or a \textbf{$b$-burst-insertion}) if exactly $b$ insertions have occurred consecutively from $\boldsymbol{x}$.
Specifically, a $b$-burst-insertion at the $i$-th position, $1\leq i \leq n+1$, will transform $x_1\cdots x_n$ into $x_1\cdots x_{i-1} \circ y_1\cdots y_b \circ x_{i}\cdots x_n \in \Sigma_q^{n+b}$, where $y_1\cdots y_b$ denotes the inserted segment of length $b$.
For any positive integer $t$, the \textbf{$b$-burst-insertion ball} of radius $t$ cantered at $\boldsymbol{x}$ is defined as $\mathcal{I}_{t,b}(\boldsymbol{x})$, representing the set of all sequences of length $n+tb$ over $\Sigma_q$ that can be obtained from $\boldsymbol{x}$ after $t$ bursts of insertions, each of length exactly $b$.
The maximum size of a $b$-burst-insertion ball of radius $t$ over all sequences in $\Sigma_q^n$ is denoted as 
\begin{align*}
  I_{q,b}(n,t) := \max\{|\mathcal{I}_{t,b}(\boldsymbol{x})|: \boldsymbol{x} \in \Sigma_q^n\}.
\end{align*}
The maximum intersection size between $b$-burst-insertion balls of radius $t$ over any two distinct sequences in $\Sigma_q^n$ is denoted as 
\begin{align*}
  N_{q,b}^+(n,t) := \max \{|\mathcal{I}_{t,b}(\boldsymbol{x}) \cap \mathcal{I}_{t,b}(\boldsymbol{y})|: \boldsymbol{x} \neq \boldsymbol{y} \in \Sigma_q^n\}.
\end{align*}

\subsection{Bursts of Deletions}
For a given sequence $\boldsymbol{x}=x_1\cdots x_n\in \Sigma_q^n$ and a positive integer $b$ with $n\geq b+1$, we say that $\boldsymbol{x}$ suffers \textbf{a burst of $b$ deletions} (or a \textbf{$b$-burst-deletion}) if exactly $b$ deletions have occurred consecutively from $\boldsymbol{x}$.
Specifically, a $b$-burst-deletion at the $i$-th position, $1\leq i \leq n-b+1$, will transform $x_1\cdots x_n$ into $x_1\cdots x_{i-1} \circ x_{i+b}\cdots x_n\in \Sigma_q^{n-b}$, where $x_i \cdots x_{i+b-1}$ denotes the deleted segment of length $b$.
When $n\geq bt+1$, where $t$ is a positive integer, the \textbf{$b$-burst-deletion ball} of radius $t$ of $\boldsymbol{x}$ is defined as $\mathcal{D}_{t,b}(\boldsymbol{x})$, representing the set of all sequences of length $n-bt$ over $\Sigma_q$ that can be obtained from $\boldsymbol{x}$ after $t$ bursts of deletions, each of length exactly $b$.
The maximum size of a $b$-burst-deletion ball of radius $t$ over all sequences in $\Sigma_q^n$ is denoted as 
\begin{align*}
  D_{q,b}(n,t) := \max\{|\mathcal{D}_{t,b}(\boldsymbol{x})|: {\boldsymbol{x} \in \Sigma_q^n} \},
\end{align*}
and the maximum intersection size between $b$-burst-deletion balls of radius $t$ over any two distinct sequences in $\Sigma_q^n$ is denoted as 
\begin{align*}
  N_{q,b}^-(n,t) := \max \{| \mathcal{D}_{t,b}(\boldsymbol{x}) \cap \mathcal{D}_{t,b}(\boldsymbol{y})|: \boldsymbol{x} \neq \boldsymbol{y} \in \Sigma_q^n\}.
\end{align*}

\subsection{Problem Descriptions}
As we explore the sequence reconstruction problem in the context of channels that introduce bursts of insertions or deletions, the objectives addressed in this paper include the following:

\begin{problem}[\textbf{Burst-Insertion Channels}]\label{pro:ins}
    Given integers $q\geq 2, b\geq 1, t\geq 1, n\geq 1$, determine:
    \begin{itemize}
        \item The minimum number of channels $N_{q,b}^+(n,t)+1$ required for unique reconstruction from $t$ burst-insertions of length $b$.
        \item A linear-time decoding algorithm  with a complexity of $O_{q,t,b}(nN)$, which, for any $N \geq N_{q,b}^+(n,t) + 1$ distinct outputs $\{\boldsymbol{y}_i\}_{i=1}^N\subseteq \mathcal{I}_{t,b}(\boldsymbol{x})$, recovers the transmitted sequence $\boldsymbol{x}$.
    \end{itemize}
\end{problem}

\begin{problem}[\textbf{Burst-Deletion Channels}]\label{pro:del}
    Given integers $q\geq 2, b\geq 1, t\geq 1, n\geq b(t+1)-1$, determine:
    \begin{itemize}
    \item The minimum number of channels $N_{q,b}^-(n,t)+1$ required for unique reconstruction from $t$ burst-deletions of length $b$.
        \item A linear-time decoding algorithm with a complexity of $O_{q,t,b}(nN)$, which, for any $N \geq N_{q,b}^-(n,t) + 1$ distinct outputs $\{\boldsymbol{y}_i\}_{i=1}^N\subseteq \mathcal{D}_{t,b}(\boldsymbol{x})$, recovers the transmitted sequence $\boldsymbol{x}$.
    \end{itemize}
\end{problem}

\subsection{Our Contributions}

Levenshtein's pioneering work \cite{Levenshtein-01-JCTA-recons} established a foundational framework for solving Problems~\ref{pro:ins} and~\ref{pro:del}, specifically under the constraint $ b = 1 $. In this paper, we extend this investigation to include the remaining cases for these problems.   

In Section \ref{sec:ins}, we offer a comprehensive solution for Problem \ref{pro:ins}. Initially, in Subsection \ref{subsec:ins_ball}, we demonstrate the cardinality identity $ |\mathcal{I}_{t,b}(\boldsymbol{x})| = I_{q,b}(n,t) $ for $\boldsymbol{x} \in \Sigma_q^n$ and derive a closed-form expression for $ I_{q,b}(n,t) $. Building upon this result, we then provide a closed-form expression for $ N_{q,b}^+(n,t) $ in Subsection \ref{subsec:ins_ball_int}. Finally, in Subsection \ref{subsec:ins_alg}, we present a reconstruction algorithm with a complexity of $ O_{q,t,b}(nN) $, which utilizes $ N \geq N_{q,b}^+(n,t) + 1 $ distinct output sequences from the channel to accurately recover the original transmitted sequence.  

In Section \ref{sec:del}, we propose a partial solution for Problem \ref{pro:del}, primarily focusing on the binary alphabet while addressing several aspects concerning non-binary alphabets. In Subsection \ref{subsec:del_ball}, we analyze the structure of burst-deletion balls and compute the quantity $ D_{q,b}(n,t) $. In Subsection \ref{subsec:del_ball_int}, we derive a closed-form expression for $ N_{2,b}^-(n, t) $ utilizing the quantity $ D_{2,b}(n,t) $, with particular emphasis on the binary case. Finally, in Subsection \ref{subsec:del_alg}, we propose a reconstruction algorithm with a complexity of $ O_{t,b}(nN) $, designed to process $ N \geq N_{2,b}^-(n,t) + 1 $ distinct output sequences from the channel to uniquely recover the original transmitted sequence, again focusing specifically on the binary case.

\section{Notations}
We define $\binom{n}{i}=0$ for $n<i$ and $\binom{n}{i}=1$ for $n=i\geq 0$.
For two integers $i$ and $j$, let $[i,j]= \{i,i+1,\ldots,j\}$ if $i \leq j$ and $[i,j]= \emptyset$ otherwise.
Assume $q\geq 2$ and $n\geq 0$, let $\Sigma_q= \{0,1,\ldots,q-1\}$ and $\Sigma_q^n$ denote the set of all sequences of length $n$ over $\Sigma_q$.
Set $\Sigma_q^0= \emptyset$.
We will write a sequence $\boldsymbol{x}\in \Sigma_q^n$ as $x_1 x_2 \cdots x_n$, where $x_i$ represents the $i$-th entry of $\boldsymbol{x}$ for $i \in [1,n]$.
For a set of indices $\mathcal{S}\subseteq [1,n]$, let $\boldsymbol{x}_{\mathcal{S}}$ be the projection of $\boldsymbol{x}$ onto $\mathcal{S}$ and we say that $\boldsymbol{x}_{\mathcal{S}}$ is a \textbf{substring} of $\boldsymbol{x}$ if $\mathcal{S}$ is an integer interval. 

For two sequences $\boldsymbol{x}$ and $\boldsymbol{y}$, their \textbf{concatenation} is denoted as $\boldsymbol{x} \circ \boldsymbol{y}$. 
Similarly, for two sets of sequences $\mathcal{S}$ and $\mathcal{T}$, their \textbf{concatenation} $\{\boldsymbol{u}\circ \boldsymbol{v}: \boldsymbol{u} \in \mathcal{S}, \boldsymbol{v} \in \mathcal{T} \}$ is denoted as $\mathcal{S} \circ \mathcal{T}$. 
In particular, when $\mathcal{S}$ only contains one sequence $\boldsymbol{u}$, we will write $\mathcal{S} \circ \mathcal{T}$ as $\boldsymbol{u} \circ \mathcal{T}$ for simplicity.
Moreover, let $\mathcal{S}^{\boldsymbol{u}}$ be the set of all sequences in $\mathcal{S}$ starting with $\boldsymbol{u}$ and $\mathcal{S}_{\boldsymbol{v}}$ be the set of all sequences in $\mathcal{S}$ ending with $\boldsymbol{v}$.
It is worth noting that for any symbol $\alpha$ and any positive integer $k$, we use the term $\alpha^k$ to denote the concatenation of $k$ copies of $\alpha$.

\section{Burst-Insertion Channels}\label{sec:ins}

In this section, we will address Problem \ref{pro:ins} for $b\geq 2$. We begin with the following easily verifiable statement regarding bursts of insertions.

\begin{claim}\label{cla:ins}
For $n, t\geq 1$, we have $\mathcal{I}_{t,b}(\boldsymbol{x})=\mathcal{I}_{t,b}(\boldsymbol{x})^{0}\sqcup \mathcal{I}_{t,b}(\boldsymbol{x})^{1} \sqcup \cdots \sqcup \mathcal{I}_{t,b}(\boldsymbol{x})^{q-1}$, where $\mathcal{I}_{t,b}(\boldsymbol{x})^{x_1}= x_1 \circ \mathcal{I}_{t,b}(\boldsymbol{x}_{[2,n]})$ and $\mathcal{I}_{t,b}(\boldsymbol{x})^{\alpha}= \alpha \circ \Sigma_q^{b-1} \circ \mathcal{I}_{t-1,b}(\boldsymbol{x})$ for $\alpha \in \Sigma_q \setminus \{x_1\}$.
\end{claim}

\subsection{The Size of a Radius-$t$ $b$-Burst-Insertion Ball}\label{subsec:ins_ball}
When $b=1$ (the case of the insertion ball), it is well known that the value of $|\mathcal{I}_{t,1}(\boldsymbol{x})|$ is a constant that is independent of the choice of $\boldsymbol{x}$. 
This raises the natural question of whether $|\mathcal{I}_{t,b}(\boldsymbol{x})|$ remains a fixed value for arbitrary $b \geq 2$. 
The case where $t=1$ has been established by Sun \textit{et al.} \cite{Sun-24-arXiv-BD}.
In what follows, we will demonstrate the validity of this result for $t\geq 1$.

\begin{theorem}\label{thm:size_ins}
For $n,t\geq 0$, $q\geq2$, $b\geq 1$, and $\boldsymbol{x}\in\Sigma_q^n$, we have $|\mathcal{I}_{t,b}(\boldsymbol{x})|= \delta_{q,b}(n,t)$, where
\begin{equation}\label{eq:ins_ball}
    \delta_{q,b}(n,t)= q^{t(b-1)}\sum_{i=0}^{t}\binom{n+t}{i}(q-1)^i.
\end{equation} 
Then, we get $I_{q,b}(n,t)= \delta_{q,b}(n,t)$, where $I_{q,b}(n,t) := \max\{|\mathcal{I}_{t,b}(\boldsymbol{x})|: \boldsymbol{x} \in \Sigma_q^n\}$.
\end{theorem}

An important observation is that $\delta_{q,b}(n,t)= q^{t(b-1)} \delta_{q,1}(n,t)$.
Moreover, it is well known that $I_{q,1}(n,t)= \delta_{q,1}(n,t)$.
To simplify our calculation, we list the following known results concerning $\delta_{q,b}(n,1)$ from the literature.

\begin{lemma}\cite[Equations (23)-(26)]{Levenshtein-01-IT-recons}\label{lem:ins_rec}
    Assume $n,t\geq 0$ and $q\geq2$, the following holds:
    \begin{itemize}
        \item $\delta_{q,1}(n,t)= |\mathcal{I}_{t,1}(\boldsymbol{x})|$ for $\boldsymbol{x}\in\Sigma_q^n$;
        \item $\delta_{q,1}(n,t)= \delta_{q,1}(n-1,t) + (q-1)\delta_{q,1}(n,t-1)$ for $n,t\geq 1$;
        \item $\delta_{q,1}(n,t)= \sum_{i=0}^{t} (q-1)^i \delta_{q,1}(n-1,t-i)$ for $n,t\geq 1$.
    \end{itemize}
\end{lemma}

By Lemma \ref{lem:ins_rec}, we can derive the following recursive formulas for $\delta_{q,b}(n,t)$ immediately.

\begin{corollary}
    For any $n\geq 1$ and $t\geq 1$, we have 
    \begin{equation}\label{eq:ins_rec}
    \begin{aligned}
        \delta_{q,b}(n,t)
        &= \delta_{q,b}(n-1,t) + (q-1) q^{b-1} \delta_{q,b}(n,t-1)\\
        &= \sum_{i=0}^{t} (q-1)^i q^{i(b-1)} \delta_{q,b}(n-1,t-i).
    \end{aligned}
    \end{equation}
\end{corollary}

\begin{IEEEproof}[Proof of Theorem \ref{thm:size_ins}]
Observe that when $n=0$, $\boldsymbol{x}$ is an empty sequence and we have $\mathcal{I}_{t,b}(\boldsymbol{x})= \Sigma_q^{bt}$, it then can be easily checked that $|\mathcal{I}_{t,b}(\boldsymbol{x})|= \delta_{q,b}(n,t)= q^{bt}$.
Below we will focus on the scenario where $n\geq 1$ and prove the theorem by induction on $n+t$.

For the base case where $n+t=1$, we have $(n,t)=(1,0)$.
In this case, there is no error occurred.
As a result, the error ball $\mathcal{I}_{t,b}(\boldsymbol{x})$ only contains $\boldsymbol{x}$ itself, implying that $|\mathcal{I}_{t,b}(\boldsymbol{x})|=1$.
Moreover, we may compute $I_{q,b}(1,0) = 1$, thereby the conclusion is valid for the base case.

Now, assuming that the conclusion is valid for $n+t<m$ with $m\geq 2$, we examine the scenario where $n+t=m$.
If $t=0$, similar to the base case, it can be easily verified that the conclusion is valid.
If $t\geq 1$, by Claim \ref{cla:ins}, we can compute 
\begin{equation}\label{eq:ins_ball'}
    \begin{aligned}
        |\mathcal{I}_{t,b}(\boldsymbol{x})|
        &= |\mathcal{I}_{t,b}(\boldsymbol{x}_{[2,n]})|+ (q-1)q^{b-1} |\mathcal{I}_{t-1,b}(\boldsymbol{x})| \\
        &\stackrel{(\ast)}{=} \delta_{q,b}(n-1,t)+(q-1)q^{b-1}\delta_{q,b}(n,t-1) \\
        &\stackrel{(\star)}{=}\delta_{q,b}(n,t),
    \end{aligned}
\end{equation}
where $(\ast)$ follows by the induction hypothesis and $(\star)$ follows from Equation (\ref{eq:ins_rec}).
Consequently, the conclusion is also valid for $n+t=m$, thereby completing the proof.
\end{IEEEproof}

For integers $n\geq 1, q\geq 2, b\geq 1, t\geq 1$, a code $\mathcal{C}\subseteq \Sigma_q^n$ is referred to as a \textbf{$t$-$b$-burst-insertion correcting code} if for any two distinct sequences $\boldsymbol{x}, \boldsymbol{y} \in \mathcal{C}$, it holds that $\mathcal{I}_{t,b}(\boldsymbol{x})\cap \mathcal{I}_{t,b}(\boldsymbol{y})= \emptyset$.
As a byproduct of Theorem \ref{thm:size_ins}, we can derive the sphere-packing upper bound on the size of $t$-$b$-burst-insertion correcting codes.

\begin{corollary}\label{cor:spb}
  For integers $n\geq 1, q\geq 2, b\geq 1, t\geq 1$, let $\mathcal{C}\subseteq \Sigma_q^n$ be a $t$-$b$-burst-insertion correcting code, then 
  \[
    |\mathcal{C}|\leq \frac{q^{n+tb}}{I_{q,b}(n,t)}= \frac{q^{n+t}}{I_{q,1}(n,t)}.
  \]
\end{corollary}

\begin{IEEEproof}
    Observe that $\mathcal{I}_{t,b}(\boldsymbol{x})\subseteq \Sigma_q^{n+tb}$ for $\boldsymbol{x}\in \Sigma_q^n$, we have
    \[
        \left|\bigcup_{\boldsymbol{x} \in \mathcal{C}} \mathcal{I}_{t,b}(\boldsymbol{x})\right|\leq q^{n+tb}.
    \]
    Moreover, since $\mathcal{C}$ is a $t$-$b$-burst-insertion correcting code, we can compute
    \[
      \left|\bigcup_{\boldsymbol{x} \in \mathcal{C}} \mathcal{I}_{t,b}(\boldsymbol{x})\right|
      = \sum_{\boldsymbol{x} \in \mathcal{C}} |\mathcal{I}_{t,b}(\boldsymbol{x})|
      = |\mathcal{C}| \times I_{q,b}(n,t).
    \]
    Then the conclusion follows.
\end{IEEEproof}

\begin{remark}
  The sphere-packing bound of $t$-$b$-burst-insertion correcting codes established in Corollary \ref{cor:spb} is independent of the parameter $b$ and is identical to that of $t$-insertion correcting codes. 
  This is quite interesting.
\end{remark}

\begin{remark}
    For integers $n\geq bt+1, q\geq 2, b\geq 1, t\geq 1$, a code $\mathcal{C}\subseteq \Sigma_q^n$ is referred to as a \textbf{$t$-$b$-burst-deletion correcting code} if for any two distinct sequences $\boldsymbol{x}, \boldsymbol{y} \in \mathcal{C}$, it holds that $\mathcal{D}_{t,b}(\boldsymbol{x})\cap \mathcal{D}_{t,b}(\boldsymbol{y})= \emptyset$.
    Similar to \cite[Lemma 1]{Levenshtein-66-SPD-1D} and \cite[Theorem 1]{Schoeny-17-IT-BD}, one can show that the set $\mathcal{C}\subseteq \Sigma_q^n$ is a $t$-$b$-burst-deletion correcting code if and only if it is a $t$-$b$-burst-insertion correcting code. This implies that the size of any $q$-ary length-$n$ $t$-$b$-burst-deletion correcting code is also upper bounded by $\frac{q^{n+t}}{I_{q,1}(n,t)}$.
\end{remark}

\subsection{The Maximum Intersection Size Between Two Radius-$t$ $b$-Burst-Insertion Balls}\label{subsec:ins_ball_int}

In the previous subsection, we showed that $I_{q,b}(n,t)$ is a multiple of $I_{q,1}(n,t)$, with the ratio $\frac{I_{q,b}(n,t)}{I_{q,1}(n,t)} = q^{t(b-1)}$. Now, we will demonstrate that a similar conclusion holds for the quantity $N_{q,b}^+(n,t) := \max \{|\mathcal{I}_{t,b}(\boldsymbol{x}) \cap \mathcal{I}_{t,b}(\boldsymbol{y})|: \boldsymbol{x} \neq \boldsymbol{y} \in \Sigma_q^n\}$, which is stated as follows.

\begin{theorem}\label{thm:size_int_ins}
For $n, t\geq 1$, we have $N_{q,b}^+(n,t)=\Delta_{q,b}(n,t)$, where 
\begin{equation}\label{eq:int}
\Delta_{q,b}(n,t):= q^{t(b-1)}\sum_{i=0}^{t-1} \binom{n+t}{i} (q-1)^i [1-(-1)^{t-i}].
\end{equation}
\end{theorem}

\begin{remark}\label{rmk:t=1}
  The case where $t=1$ has been established in \cite[Theorem 3.4]{Sun-23-IT-BDR}.
  Moreover, \cite[Claim 3.3]{Sun-23-IT-BDR} stated that the maximum ball intersection size $N_{q,b}^+(n,t)=2q^{b-1}$ can be achieved when the selected two centers are one Hamming distance apart.
\end{remark}

It can be easily checked that $\Delta_{q,b}(n,t)= q^{t(b-1)} \Delta_{q,1}(n,t)$.
Similar to the previous subsection, we list the following known results concerning $\Delta_{q,1}(n,t)$ from the literature.

\begin{lemma}\cite[Equations (51)-(56)]{Levenshtein-01-IT-recons}\label{lem:ins_int_rec}
    Assume $n, t\geq 1$ and $q\geq2$, the following holds:
    \begin{itemize}
        \item $\Delta_{q,1}(n-1,t)+\Delta_{q,1}(n,t-1)= 2I_{q,1}(n,t-1)$;
        \item $\Delta_{q,1}(n,t)= \Delta_{q,1}(n-1,t) + (q-1)\Delta_{q,b}(n,t-1)= \sum_{i=0}^{t} (q-1)^i \Delta_{q,1}(n-1,t-i)$;
        \item $\Delta_{q,1}(n,t)= 2I_{q,1}(n,t-1)+ (q-2)\Delta_{q,1}(n,t)= 2\sum_{i=1}^t (q-2)^{i-1}I_{q,1}(n,t-i)$.
    \end{itemize}
\end{lemma}

By combining Lemma \ref{lem:ins_int_rec} with the ratio $\frac{I_{q,b}(n,t)}{I_{q,1}(n,t)}= \frac{\Delta_{q,b}(n,t)}{\Delta_{q,1}(n,t)} = q^{t(b-1)}$, we can develop the following recursive formulas for $\Delta_{q,b}(n,t)$ immediately.

\begin{lemma}\label{lem:pro_int_ins}
For $n, t\geq 1$, the following equations hold:
\begin{equation}\label{eq:ins_int&ins}
    \Delta_{q,b}(n-1,t) + q^{b-1}\Delta_{q,b}(n,t-1) = 2q^{b-1}I_{q,b}(n,t-1);
\end{equation}
\begin{equation}\label{eq:ins_int_rec}
    \begin{aligned}
        \Delta_{q,b}(n,t) 
        &= \Delta_{q,b}(n-1,t) + (q-1)q^{b-1}\Delta_{q,b}(n,t-1) \\
        &= \sum_{i=0}^{t-1} (q-1)^i q^{i(b-1)} \Delta_{q,b}(n-1,t-i);
    \end{aligned}
\end{equation}
\begin{equation}\label{eq:ins_int&ins'}
    \begin{aligned}
    \Delta_{q,b}(n,t) 
    &= 2q^{b-1}I_{q,b}(n,t-1) + (q-2)q^{b-1}\Delta_{q,b}(n,t-1) \\
    &= 2\sum_{i=1}^t (q-2)^{i-1} q^{i(b-1)} I_{q,b}(n,t-i).
    \end{aligned}
\end{equation}
\end{lemma}

Specially, when $q=2$, we can establish the following alternative expression for $\Delta_{q,b}(n,t)$ using Equations (\ref{eq:ins_ball}), (\ref{eq:ins_int&ins}), and (\ref{eq:ins_int_rec}).

\begin{corollary}
    For $n, t\geq 1$, we have $\Delta_{2,b}(n,t) = 2^b I_{2,b}(n,t-1) = 2^{t(b-1)+1}\sum_{i=0}^{t-1}\binom{n+t-1}{i}$.
\end{corollary}

\begin{IEEEproof}[Proof of Theorem \ref{thm:size_int_ins}]
We first show that $N_{q,b}^+(n,t)\geq \Delta_{q,b}(n,t)$ by finding two sequences such that the intersection size between their $b$-burst-insertion balls of radius $t$ equals $\Delta_{q,b}(n,t)$.
For any $\boldsymbol{x}=x_1x_2 \cdots x_n\in\Sigma_q^n$, let $\boldsymbol{y} \in \Sigma_q^n$ be such that $y_1\neq x_1$ and $y_i=x_i$ for $i\in [2,n]$. We claim that $| \mathcal{I}_{t,b}(\boldsymbol{x}) \cap \mathcal{I}_{t,b}(\boldsymbol{y}) |=\Delta_{q,b}(n,t)$ and demonstrate its correctness by induction on $n+t$.
For the base case where $n+t=2$, we have $t=1$. Then the conclusion follows by Remark \ref{rmk:t=1}.
Now, assuming that the conclusion is valid for $n+t<m$ with $m\geq 3$, we examine the scenario where $n+t=m$.
By Claim \ref{cla:ins}, we can obtain 
\begin{equation*}
  \mathcal{I}_{t,b}(\boldsymbol{x})^{\alpha}= 
  \begin{cases}
    x_1\circ \mathcal{I}_{t,b}(\boldsymbol{x}_{[2,n]}), & \mbox{if } \alpha= x_1; \\
    \alpha \circ \Sigma_q^{b-1} \circ \mathcal{I}_{t-1,b}(\boldsymbol{x}), & \mbox{if } \alpha \in \Sigma_q \setminus \{x_1\},
  \end{cases}
\end{equation*}
and 
\begin{equation*}
  \mathcal{I}_{t,b}(\boldsymbol{y})^{\alpha}= 
  \begin{cases}
    y_1\circ \mathcal{I}_{t,b}(\boldsymbol{x}_{[2,n]}), & \mbox{if } \alpha= y_1; \\
    \alpha \circ \Sigma_q^{b-1} \circ \mathcal{I}_{t-1,b}(\boldsymbol{y}), & \mbox{if } \alpha \in \Sigma_q \setminus \{y_1\}.
  \end{cases}
\end{equation*}
Observe that 
\begin{gather*}
  \Sigma_q^{b-1}\circ \mathcal{I}_{t-1,b}(\boldsymbol{y})\subseteq \mathcal{I}_{t,b}(\boldsymbol{y}_{[2,n]})= \mathcal{I}_{t,b}(\boldsymbol{x}_{[2,n]}), \\
  \Sigma_q^{b-1}\circ \mathcal{I}_{t-1,b}(\boldsymbol{x})\subseteq \mathcal{I}_{t,b}(\boldsymbol{x}_{[2,n]})= \mathcal{I}_{t,b}(\boldsymbol{y}_{[2,n]}).
\end{gather*}
Let $\mathcal{S}:= \mathcal{I}_{t,b}(\boldsymbol{x})\cap \mathcal{I}_{t,b}(\boldsymbol{y})$, it follows that
\begin{equation*}
  \mathcal{S}^{\alpha}= 
  \begin{cases}
    x_1 \circ \Sigma_q^{b-1}\circ \mathcal{I}_{t-1,b}(\boldsymbol{y}), & \mbox{if } \alpha= x_1; \\
    y_1 \circ \Sigma_q^{b-1}\circ \mathcal{I}_{t-1,b}(\boldsymbol{x}), & \mbox{if } \alpha= y_1;\\
    \alpha \circ \Sigma_q^{b-1}\circ \big( \mathcal{I}_{t-1,b}(\boldsymbol{x}) \cap \mathcal{I}_{t-1,b}(\boldsymbol{y}) \big), & \mbox{if } \alpha \in \Sigma_q\setminus \{x_1,y_1\}.
  \end{cases}
\end{equation*}
Then we can compute
\begin{align*}
    |\mathcal{S}|
    &= |\mathcal{S}^{x_1}|+ |\mathcal{S}^{y_1}|+ \sum_{\alpha \in \Sigma_q \setminus \{x_1,y_1\}} |\mathcal{S}^{\alpha}| \\
    &= q^{b-1} |\mathcal{I}_{t-1,b}(\boldsymbol{y})|+ q^{b-1} |\mathcal{I}_{t-1,b}(\boldsymbol{x})|+ (q-2)q^{b-1} |\mathcal{I}_{t-1,b}(\boldsymbol{x}) \cap \mathcal{I}_{t-1,b}(\boldsymbol{y})|\\
    &\stackrel{(\ast)}{=} 2q^{b-1}I_{q,b}(n,t-1)+ (q-2)q^{b-1} \Delta_{q,b}(n,t-1) \\
    &\stackrel{(\star)}{=}\Delta_{q,b}(n,t),
\end{align*}
where $(\ast)$ follows by the induction hypothesis and $(\star)$ follows by Equation (\ref{eq:ins_int&ins'}).
Consequently, the conclusion is also valid for $n+t=m$, thereby $| \mathcal{I}_{t,b}(\boldsymbol{x}) \cap \mathcal{I}_{t,b}(\boldsymbol{y}) |=\Delta_{q,b}(n,t)$. Then we get $N_{q,b}^+(n,t)\geq \Delta_{q,b}(n,t)$.

To complete the proof, it remains to show that $N_{q,b}^+(n,t)\leq\Delta_{q,b}(n,t)$. We demonstrate its correctness again by induction on $n+t$.
For the base case where $n+t=2$, we have $t=1$. Then the conclusion follows by Remark \ref{rmk:t=1}.
Now, assuming the conclusion is valid for $n+t<m$, we examine the scenario where $n+t=m\geq 3$. For any two distinct sequences $\boldsymbol{x}, \boldsymbol{y} \in\Sigma_q^{n}$, let $\mathcal{S}:= \mathcal{I}_{t,b}(\boldsymbol{x})\cap \mathcal{I}_{t,b}(\boldsymbol{y})$.
We will distinguish between two cases based on whether $x_1=y_1$.
\begin{itemize}
    \item If $x_1=y_1$, by Claim \ref{cla:ins}, we can obtain
    \begin{equation*}
      \mathcal{S}^{\alpha}= 
      \begin{cases}
        x_1 \circ \big( \mathcal{I}_{t,b}(\boldsymbol{x}_{[2,n]}) \cap \mathcal{I}_{t,b}(\boldsymbol{y}_{[2,n]})\big), & \mbox{if } \alpha= x_1; \\
        \alpha \circ \Sigma_q^{b-1} \circ \big( \mathcal{I}_{t-1,b}(\boldsymbol{x}) \cap \mathcal{I}_{t-1,b}(\boldsymbol{y})\big), & \mbox{if } \alpha \in \Sigma_q \setminus \{x_1\}.
      \end{cases}
    \end{equation*}
    Then we can compute 
    \begin{align*}
        |\mathcal{S}|
        &= |\mathcal{S}^{x_1}|+ \sum_{\alpha \in \Sigma_q \setminus \{x_1\}} |\mathcal{S}^{\alpha}| \\
        &\stackrel{(\ast)}{\leq} \Delta_{q,b}(n-1,t) + (q-1)q^{b-1}\Delta_{q,b}(n,t-1) \\
        &\stackrel{(\star)}{=} \Delta_{q,b}(n,t),
    \end{align*}
    where $(\ast)$ follows by the induction hypothesis and $(\star)$ follows by Equation (\ref{eq:ins_int_rec}).
    
    \item If $x_1\neq y_1$, again by Claim \ref{cla:ins}, we can obtain
    \begin{equation*}
      \mathcal{S}^{\alpha}\subseteq 
      \begin{cases}
        \mathcal{I}_{t,b}(\boldsymbol{y})^{x_1}
        \subseteq x_1 \circ \Sigma_q^{b-1}\circ \mathcal{I}_{t-1,b}(\boldsymbol{y}), & \mbox{if } \alpha= x_1; \\
        \mathcal{I}_{t,b}(\boldsymbol{x})^{y_1}
        \subseteq y_1 \circ \Sigma_q^{b-1}\circ \mathcal{I}_{t-1,b}(\boldsymbol{x}), & \mbox{if } \alpha= y_1;\\
        \alpha \circ \Sigma_q^{b-1}\circ \big( \mathcal{I}_{t-1,b}(\boldsymbol{x}) \cap \mathcal{I}_{t-1,b}(\boldsymbol{y}) \big), & \mbox{if } \alpha \in \Sigma_q\setminus \{x_1,y_1\}.
      \end{cases}
    \end{equation*}
    Then we can compute
    \begin{align*}
    |\mathcal{S}|
    &= |\mathcal{S}^{x_1}|+ |\mathcal{S}^{y_1}|+ \sum_{\alpha \in \Sigma_q \setminus \{x_1,y_1\}} |\mathcal{S}^{\alpha}| \\
    &\leq q^{b-1} |\mathcal{I}_{t-1,b}(\boldsymbol{y})|+ q^{b-1} |\mathcal{I}_{t-1,b}(\boldsymbol{x})|+ (q-2)q^{b-1} |\mathcal{I}_{t-1,b}(\boldsymbol{x}) \cap \mathcal{I}_{t-1,b}(\boldsymbol{y})|\\
    &\stackrel{(\ast)}{=} 2q^{b-1}I_{q,b}(n,t-1)+ (q-2)q^{b-1} \Delta_{q,b}(n,t-1) \\
    &\stackrel{(\star)}{=}\Delta_{q,b}(n,t),
    \end{align*}
    where $(\ast)$ follows by the induction hypothesis and $(\star)$ follows by Equation (\ref{eq:ins_int&ins'}).
\end{itemize}
Consequently, the conclusion is also valid for $n+t=m$, thereby completing the proof.
\end{IEEEproof}

\subsection{The Reconstruction Algorithm for Burst-Insertion Errors} \label{subsec:ins_alg}

We will now design a straightforward algorithm, presented in Algorithm \ref{alg:ins}, to recover an arbitrary $\boldsymbol{x} \in \Sigma_q^n$, where $n \geq 1$, given its length $n$ and $N \geq N_{q,b}^+(n, t) + 1$ distinct sequences $\boldsymbol{y}_1, \boldsymbol{y}_2, \ldots, \boldsymbol{y}_N$ in $\mathcal{I}_{t,b}(\boldsymbol{x})$.
This algorithm generalizes the threshold reconstruction algorithm presented in \cite[Section 6]{Levenshtein-01-JCTA-recons}.
The main challenge is to determine appropriate thresholds, which we present in Lemmas \ref{lem:ins_x1} and \ref{lem:ins_number}.

\begin{algorithm}
    \caption{Sequence Reconstruction Algorithm under Burst-Insertion Channel}\label{alg:ins}
    \KwIn{$\mathcal{U} \subseteq \mathcal{I}_{t,b}(\boldsymbol{x})$ for some $\boldsymbol{x}\in \Sigma_q^n$ with $|\mathcal{U}|\geq N_{q,b}^+(n,t)+1$}
    \KwOut{$\boldsymbol{x}= x_1\cdots x_n$}

    \textbf{Initialization:} Set $i=0$\\
        \While{$i<n$ or $t>0$}
            {$i \leftarrow i+1$\\
             Let $\beta \in \Sigma_q$ be such that $|\mathcal{U}^{\alpha \succ \beta,b}|< |\mathcal{U}^{\beta \succ \alpha,b}|$ for any $\alpha\in \Sigma_q\setminus \{\beta\}$\\
             $x_i \leftarrow \beta$\\
             Let $j\in [0,t]$ be the largest index such that $|\mathcal{U}^{x_i,j+1,b}|\geq (q-1)^j q^{j(b-1)} N_{q,b}^+(n-i,t-j)+1$\\
             $\mathcal{U}\leftarrow \overline{\mathcal{U}^{x_i,j+1,b}}$\\
             \eIf{$j=t$}
                {\{$x_{i+1}\cdots x_n\} \leftarrow \mathcal{U}$}
                {$t\leftarrow t-j$}
            }
\end{algorithm}

The algorithm is operated as follows:
\begin{itemize}
  \item Firstly, we define the set $\mathcal{U} = \{ \boldsymbol{y}_1, \boldsymbol{y}_2, \ldots, \boldsymbol{y}_N \}$, which can be considered as a matrix of size $(n + bt) \times N$, with the vector in the $i$-th column set as $\boldsymbol{y}_i$ for $i \in [1, N]$. 
  \item We then use Lemma \ref{lem:ins_x1} to recover $x_1$ from the first $tb + 1$ rows of this matrix. This completes the recovery of $\boldsymbol{x}$ if $n = 1$. 
  \item In the case where $n \geq 2$, we will apply Lemma \ref{lem:ins_number} to identify a submatrix $\mathcal{U}'$ of the matrix $\mathcal{U}$, which has size $(n' + bt') \times N'$ and consists of distinct columns $\boldsymbol{y}' \in \mathcal{I}_{t',b}(\boldsymbol{x}')$, where $\boldsymbol{x}' = x_2 \cdots x_n$, $n' = n - 1$, $t' = t - j$, and $N' \geq N_{q,b}^+(n', t') + 1$.
  \item If $t' = t - j = 0$, then $N_{q,b}^+(n', t')=0$, implying that the matrix $\mathcal{U}'$ consists of only the column $\boldsymbol{x}'$, and recovering $\boldsymbol{x} = x_1 \circ \boldsymbol{x}'$ is completed. 
  \item In the case where $t' = t - j \geq 1$, recovering $\boldsymbol{x} \in \Sigma_q^n$ reduces to recovering $\boldsymbol{x}' \in \Sigma_q^{n'}$ with the help of $N' \geq N_{q,b}^+(n', t') + 1$ distinct sequences of length $n' + bt'$ in $\mathcal{I}_{t',b}(\boldsymbol{x}')$. 
  \item Since $n' = n - 1$ and $t' \leq t$, we will eventually reach a stage where $n = 1$ or $t = 0$. At this point, we can derive the correct $\boldsymbol{x}$.
\end{itemize}

We begin with the following definitions.

\begin{definition}
  Assume $n,t\geq 1$. For any $\mathcal{U}\subseteq \Sigma_q^{n+bt}$, $\alpha \in \Sigma_q$, and $i\in [1,t+1]$, we define $\mathcal{U}^{\alpha,i,b}$ as the set consisting of all the sequences in $\mathcal{U}$ whose first appearance of the symbol $\alpha$ is at position $(i-1)b+1$ among the indices $1, b+1, \ldots, tb+1$. In other words, 
  \begin{align*}
    \mathcal{U}^{\alpha,i,b}=\{\boldsymbol{y}\in \mathcal{U}: y_{(i-1)b+1}=\alpha \text{ and } y_{(j-1)b+1}\neq \alpha \text{ for any } j\in [1,i-1]\}.
  \end{align*}
  Since sequences in $\mathcal{U}^{\alpha,i,b}$ have $(q-1)^{i-1} q^{(i-1)(b-1)}$ choices for the first $(i-1)b+1$ entries, we can partition $\mathcal{U}^{\alpha,i,b}$ into $(q-1)^{i-1} q^{(i-1)(b-1)}$ disjoint subsets.
  We then arbitrarily choose one subset with the largest cardinality among these subsets and denote it as $\underline{\mathcal{U}^{\alpha,i,b}}$. 
  Finally, we define $\overline{\mathcal{U}^{\alpha,i,b}}$ as the set of sequences that can be obtained from $\underline{\mathcal{U}^{\alpha,i,b}}$ by deleting the first $(i-1)b+1$ entries. That is, 
  \begin{align*}
    \overline{\mathcal{U}^{\alpha,i,b}}=\{\boldsymbol{y}_{[(i-1)b+2,n]}: \boldsymbol{y} \in \underline{\mathcal{U}^{\alpha,i,b}}\}.
  \end{align*}
\end{definition}

\begin{definition}
  Assume $n,t\geq 1$. For any $\mathcal{U}\subseteq \Sigma_q^{n+bt}$, any index $i\in [1,t+1]$, and any two distinct symbols $\alpha, \beta\in \Sigma_q$, we define $\mathcal{U}^{\alpha \succ \beta,i,b}$ as the set consisting of all the sequences in $\mathcal{U}$ whose first appearance of the symbol $\alpha$ is $(i-1)b+1$ and $\alpha$ appears before $\beta$ among the indices $1, b+1, \ldots, tb+1$, i.e., $\mathcal{U}^{\alpha\succ \beta,i,b}=\{\boldsymbol{y}\in \mathcal{U}: y_{(i-1)b+1}=\alpha \text{ and } y_{(j-1)b+1}\not\in \{\alpha,\beta\} \text{ for any } j\in [1,i-1]\}$.
  Then we set $\mathcal{U}^{\alpha\succ \beta,b}= \cup_{i=1}^{t+1} \mathcal{U}^{\alpha\succ \beta,i,b}$.
\end{definition}

\begin{example}
  Assume $n=t=q=b=2$.
  Let $\mathcal{U}= \{010000, 010100, 010101, 100001, 101001,101011\}$, then we have
  \begin{itemize}
    \item $\{1,b+1,\ldots,tb+1\}$= $\{1,3,5\}$;
    \item $\mathcal{U}^{0,1,b}= \{010000, 010100, 010101\}$, $\mathcal{U}^{0,2,b}= \{100001\}$, $\mathcal{U}^{0,3,b}= \{101001\}$, $\mathcal{U}^{1,0,b}= \{100001,101001,101011\}$, and $\mathcal{U}^{1,1,b}=\mathcal{U}^{1,2,b}=\emptyset$;
    \item $\overline{\mathcal{U}^{0,1,b}}= \{10000, 10100, 10101\}$, $\overline{\mathcal{U}^{0,2,b}}= \{001\}$, $\overline{\mathcal{U}^{0,3,b}}= \{1\}$, $\overline{\mathcal{U}^{1,0,b}}= \{00001,01001,01011\}$, and $\overline{\mathcal{U}^{1,1,b}}= \overline{\mathcal{U}^{1,2,b}} =\emptyset$;
    \item $\mathcal{U}^{0\succ 1,1,b}= \{010000, 010100, 010101\}$, $\mathcal{U}^{0\succ 1,2,b}= \mathcal{U}^{0\succ 1,3,b}= \emptyset$, $\mathcal{U}^{1\succ 0,1,b}= \{100001, 101001,101011\}$, and $\mathcal{U}^{1\succ 0,2,b}=\mathcal{U}^{1\succ 0,3,b}= \emptyset$;
    \item $\mathcal{U}^{0\succ 1,b}= \{010000, 010100, 010101\}$ and $\mathcal{U}^{1\succ 0,b}= \{100001, 101001,101011\}$.
  \end{itemize}
\end{example}

\begin{observation}\label{obs:ins_int}
    For any $\boldsymbol{x}\in \Sigma_q^n$ with $n\geq 1$ and $\boldsymbol{y}\in \mathcal{I}_{t,b}(\boldsymbol{x})$ with $t,b\geq 1$, we have $y_{ib+1}=x_1$ for some $i\in [0,t]$. 
    Therefore, for any $\mathcal{U}\subseteq \mathcal{I}_{t,b}(\boldsymbol{x})$, we can partition $\mathcal{U}$ into the following $t+1$ disjoint subsets: $\mathcal{U}^{x_1,1,b}, \mathcal{U}^{x_1,2,b}, \ldots, \mathcal{U}^{x_1,t+1,b}$.
    As a result, we have $|\mathcal{U}|= \sum_{i=1}^{t+1} |\mathcal{U}^{x_1,i,b}|$.
    Moreover, for any $\alpha \in \Sigma_q \setminus\{x_1\}$, we can also partition $\mathcal{U}$ into the following two disjoint subsets: $\mathcal{U}^{\alpha \succ x_1,b}, \mathcal{U}^{x_1 \succ \alpha,b}$.
\end{observation}

The following lemma shows that we can always determine the first symbol of $\boldsymbol{x}$ when we have received more than $N_{q,b}^+(n,t)$ distinct sequences derived from $\boldsymbol{x}$ after $t$ bursts of insertions, each of length exactly $b$.

\begin{lemma}\label{lem:ins_x1}
    Let $n,t\geq 1$ and $q,b\geq 2$. Given $\boldsymbol{x}\in \Sigma_q^n$ and $\mathcal{U}\subseteq \mathcal{I}_{t,b}(\boldsymbol{x})$.
    For any $\alpha\in \Sigma_q\setminus \{x_1\}$, we have $|\mathcal{U}^{\alpha \succ x_1,b}|\leq \frac{1}{2}N_{q,b}^+(n,t)$.
    As a result, when $|\mathcal{U}|\geq N_{q,b}^+(n,t)+1$, we have $|\mathcal{U}^{\alpha \succ x_1,b}|< |\mathcal{U}^{x_1 \succ \alpha,b}|$.
\end{lemma}

\begin{IEEEproof}
    By Observation \ref{obs:ins_int}, we can partition $\mathcal{U}^{\alpha \succ x_1,b}$ into the following $t$ disjoint subsets: 
    \begin{align*}
      \mathcal{U}^{\alpha \succ x_1,1,b}, \mathcal{U}^{\alpha \succ x_1,2,b}, \ldots, \mathcal{U}^{\alpha \succ x_1,t,b}.
    \end{align*}
    Let $\mathcal{T} := \big(\Sigma_q \setminus \{\alpha,x_1\}\big) \circ \Sigma_q^{b-1}$. 
    For any $i\in [1,t]$, it can be easily checked that
    \begin{align*}
      \mathcal{U}^{\alpha \succ x_1,i,b}
      &\subseteq \underbrace{\mathcal{T} \circ \mathcal{T} \circ \cdots \circ \mathcal{T}}_{i-1} \circ~ \alpha \circ \Sigma_q^{b-1} \circ \mathcal{I}_{t-i,b}(\boldsymbol{x}).
    \end{align*}
    It follows by Theorem \ref{thm:size_ins} that
    \begin{align*}
      |\mathcal{U}^{\alpha \succ x_1,i,b}|
      &\leq \left[(q-2)q^{b-1}\right]^{i-1} q^{b-1} I_{q,b}(n,t-i) \\
      &= (q-2)^{i-1}q^{i(b-1)}I_{q,b}(n,t-i).
    \end{align*}
    Then we can compute
    \begin{align*}
      |\mathcal{U}^{\alpha \succ x_1,b}|
      &= \sum_{i=1}^t |\mathcal{U}^{\alpha \succ x_1,i,b}|\\
      &\leq \sum_{i=1}^t (q-2)^{i-1}q^{i(b-1)}I_{q,b}(n,t-i) \\
      &= \frac{1}{2}N_{q,b}^+(n,t),
    \end{align*}
    where the last equality follows by Theorem \ref{thm:size_int_ins} and Equation (\ref{eq:ins_int&ins'}), thereby completing the proof.
\end{IEEEproof}

Once we determine the first symbol $x_1$, we need to identify a certain number $j\in [0,t]$ of bursts of insertions that were occurred before $x_1$ such that the submatrix $\overline{\mathcal{U}^{x_1,j+1,b}}$ satisfies the condition $|\overline{\mathcal{U}^{x_1,j+1,b}}|\geq N_{q,b}^+(n-1,t-j)+1$. This can be achieved by applying Theorem \ref{thm:size_int_ins} and Equation (\ref{eq:ins_int_rec}), as demonstrated below.

\begin{lemma}\label{lem:ins_number}
    Let $n,t\geq 1$ and $q,b\geq 2$. Given $\boldsymbol{x}\in \Sigma_q^n$ and $\mathcal{U}\subseteq \mathcal{I}_{t,b}(\boldsymbol{x})$ with $|\mathcal{U}|\geq N_{q,b}^+(n,t)+1$, then $\overline{\mathcal{U}^{x_1,i,b}}\subseteq \mathcal{I}_{t-i+1,b}(\boldsymbol{x}_{[2,n]})$ for $i\in [1,t+1]$ and there exists some $j \in [0,t]$ such that $|\overline{\mathcal{U}^{x_1,j+1,b}}|\geq N_{q,b}^+(n-1,t-j)+1$.
    In particular, if $j=t$, $\overline{\mathcal{U}^{x_1,j+1,b}}$ consists of the only sequence $x_2\cdots x_n$.
\end{lemma}

\begin{IEEEproof}
    Let $\mathcal{T} := \big(\Sigma_q \setminus \{x_1\}\big) \circ \Sigma_q^{b-1}$. 
    For any $i\in [1,t+1]$, it can be easily checked that
    \begin{align*}
      \mathcal{U}^{x_1,i,b}
      &\subseteq \underbrace{\mathcal{T} \circ \mathcal{T} \circ \cdots \circ \mathcal{T}}_{i-1} \circ~ x_1 \circ \mathcal{I}_{t-i+1,b}(\boldsymbol{x}_{[2,n]}).
    \end{align*}
    This implies that $\overline{\mathcal{U}^{x_1,i,b}}\subseteq \mathcal{I}_{t-i+1,b}(\boldsymbol{x}_{[2,n]})$ for $i\in [1,t+1]$.
    Since by Observation \ref{obs:ins_int} we have $|\mathcal{U}|= \sum_{i=1}^{t+1} |\mathcal{U}^{x_1,i,b}|$, by Theorem \ref{thm:size_int_ins} and Equation (\ref{eq:ins_int_rec}), there exists some $j\in [0,t]$ such that $|\mathcal{U}^{x_1,j+1,b}|\geq (q-1)^{j} q^{j(b-1)} N_{q,b}^+(n-1,t-j)+1$.
    It then follows by the definition of $\overline{\mathcal{U}^{x_1,j+1,b}}$ and the pigeonhole principle that $|\overline{\mathcal{U}^{x_1,j+1,b}}|\geq N_{q,b}^+(n-1,t-j)+1$, thereby completing the proof.
\end{IEEEproof}

Now we are ready to present our main contribution in this subsection.

\begin{theorem}
  Let $n,t\geq 1$ and $q,b\geq 2$. Given $\boldsymbol{x}\in \Sigma_q^n$ and $\mathcal{U}\subseteq \mathcal{I}_{t,b}(\boldsymbol{x})$ with $|\mathcal{U}|=N\geq N_{q,b}^+(n,t)+1$, Algorithm \ref{alg:ins} outputs the correct $\boldsymbol{x}$ and runs in $O_{q,t,b}(nN)$ time.
\end{theorem}

\begin{IEEEproof}
Clearly, the correctness of Algorithm \ref{alg:ins} is supported by Lemmas \ref{lem:ins_x1} and \ref{lem:ins_number}. 
In each iteration of the algorithm, specifically in Steps 3, 6, and 7, we need to calculate $|\mathcal{U}^{\alpha \succ \beta, b}|$ for $\alpha \neq \beta \in \Sigma_q$, compute $|\mathcal{U}^{x_i, j+1, b}|$, and determine $\overline{\mathcal{U}^{x_i, j+1, b}}$, respectively, from a matrix with $(tb+1)$ rows and at most $N$ columns. This process runs in $O_{q,t,b}(N)$ time. 
Since the algorithm iterates at most $n$ times, the total complexity is $O_{q,t,b}(nN)$, thereby completing the proof.

\end{IEEEproof}

\section{Burst-Deletion Channels}\label{sec:del}

In this section, we will address Problem \ref{pro:del} for $b\geq 2$ and $q=2$.
It should be noted that most results are proven for $q\geq 2$.
We begin with a few easily verifiable statements regarding bursts of deletions.

\begin{claim}\label{cla:del}
Assume $n\geq tb+1$ where $t,b\geq 1$, we consider the $b$-burst-deletion ball of radius $t$ centered at $\boldsymbol{x}\in \Sigma_q^n$, i.e.,  $\mathcal{D}_{t,b}(\boldsymbol{x})$, and the following statements are true:
\begin{itemize}
    \item $\mathcal{D}_{t,b}(\boldsymbol{x}) \subseteq \mathcal{D}_{t+1,b}(\boldsymbol{y})$ for any $\boldsymbol{y} \in \mathcal{I}_{1,b}(\boldsymbol{x})$;

    
    \item $\mathcal{D}_{t,b}(\boldsymbol{x})= \mathcal{D}_{t,b}(\boldsymbol{x})^0 \sqcup \mathcal{D}_{t,b}(\boldsymbol{x})^1 \sqcup \cdots \sqcup \mathcal{D}_{t,b}(\boldsymbol{x})^{q-1}$ and $\mathcal{D}_{t,b}(\boldsymbol{x})^{\alpha}= \alpha \circ \mathcal{D}_{t-i,b}(\boldsymbol{x}_{[ib+2,n]})$ for any $\alpha\in \Sigma_q$, where $i\geq 0$ represents the smallest index such that $x_{ib+1}=\alpha$;

    \item $\mathcal{D}_{t,b}(\boldsymbol{x})= \mathcal{D}_{t,b}(\boldsymbol{x})_0 \sqcup \mathcal{D}_{t,b}(\boldsymbol{x})_1 \sqcup \cdots \sqcup \mathcal{D}_{t,b}(\boldsymbol{x})_{q-1}$ and $\mathcal{D}_{t,b}(\boldsymbol{x})_{\alpha}= \mathcal{D}_{t-i,b}(\boldsymbol{x}_{[1,n-ib-2]})\circ \alpha$ for any $\alpha\in \Sigma_q$, where $i\geq 0$ represents the smallest index such that $x_{n-ib-1}=\alpha$.
\end{itemize}
\end{claim}

\subsection{The Size of a Radius-$t$ $b$-Burst-Deletion Ball}\label{subsec:del_ball}

When $b=1$ (the case of the deletion ball), it is well known that the value of $|\mathcal{D}_{t,1}(\boldsymbol{x})|$ depends on the choice of $\boldsymbol{x}$ and is strongly correlated with the number of runs in $\boldsymbol{x}$.
A \textbf{run} in $\boldsymbol{x}$ is a substring of $\boldsymbol{x}$ consisting of the same symbol with maximal length.
Intuitively, the size of $\mathcal{D}_{t,b}(\boldsymbol{x})$ also depends on the choice of $\boldsymbol{x}$.
This leads to the natural question: what factors determine the value of $|\mathcal{D}_{t,b}(\boldsymbol{x})|$ for arbitrary $b\geq 2$?
We first consider the simplest case where $t=1$ to identify the quantity that determines the size of $\mathcal{D}_{1,b}(\boldsymbol{x})$.
Subsequently, we will demonstrate that this quantity is also strongly connected to the maximum size of a $b$-burst-deletion ball of radius $t$ for arbitrary $t\geq 2$.

When $t=1$, Levenshtein \cite{Levenshtein-70-BD} derived an expression of $|\mathcal{D}_{1,b}(\boldsymbol{x})|$ by viewing $\boldsymbol{x}$ as an array with $b$ rows and using the number of runs in each row, which we present as follows.

\begin{lemma}\cite{Levenshtein-70-BD}\label{lem:Levenshtein}
    Assume $n\geq b+1$ with $b\geq 1$. For any $\boldsymbol{x}\in \Sigma_q^n$, we define the \textbf{$b\times \lceil \frac{n}{b} \rceil$ array representation} of $\boldsymbol{x}$ by $A_b(\boldsymbol{x})$, which can be obtained from $\boldsymbol{x}$ by arranging it in a $b\times \lceil \frac{n}{b} \rceil$ array column by column.
    Note that the last column may not be fully filled, we will fill in the blank space by repeating the last symbol in the corresponding row.
    For $i\in [1,b]$, let $A_{b,i}(\boldsymbol{x})$ be the $i$-th row of $A_b(\boldsymbol{x})$ and let $r(A_{b,i}(\boldsymbol{x}))$ be the number of runs in $A_{b,i}(\boldsymbol{x})$, then \begin{align*}
        |\mathcal{D}_{1,b}(\boldsymbol{x})|= 1+ \sum_{i=1}^b \Big(r\big(A_{b,i}(\boldsymbol{x})\big)-1\Big).
    \end{align*}
\end{lemma}

Observe that for any $\boldsymbol{x}\in \Sigma_q^n$, we have $r(\boldsymbol{x})= 1+ |\{j \in [2,n]: x_j\neq x_{j-1}\}|$. Then by Lemma \ref{lem:Levenshtein}, we can compute 
\begin{equation}\label{eq:del_ball}
    \begin{aligned}
    |\mathcal{D}_{1,b}(\boldsymbol{x})|
    &= 1+ \sum_{i=1}^b |\{ j\in [2,\lceil n/b \rceil]: x_{(j-1)b+i}\neq x_{(j-2)b+i} \}|\\
    &= 1+ |\{j\in [b+1,n]: x_j\neq x_{j-b}\}|.
\end{aligned}
\end{equation}
Indeed, Equation (\ref{eq:del_ball}) generalizes the expression for the case where $b=1$.
To this end, we introduce the concept of a $b$-run in a sequence, which serves as a natural generalization of a run.
Assume $n\geq b$, we say that $\boldsymbol{x}$ is \textbf{$b$-run sequence} if $x_i=x_{i-b}$ for any $i\in [b+1,n]$.
Furthermore, for any interval $\mathcal{K}\subseteq [1,n]$ with $|\mathcal{K}|\geq b$, we say that $\boldsymbol{x}_{\mathcal{K}}$ is a \textbf{$b$-run} of $\boldsymbol{x}$ if there is no interval $ \mathcal{K}' \supset \mathcal{K}$ such that $\boldsymbol{x}_{\mathcal{K}'}$ is a $b$-run sequence. 
From this perspective, we can provide an alternative proof for Equation (\ref{eq:del_ball}).

\begin{lemma}\cite[Claim 3.1]{Sun-23-IT-BDR}\label{lem:del_ball}
    Let $\mathcal{K}= [k_1,k_2] \subseteq [1,n]$ with $k_2-k_1\geq b\geq 1$ and $n\geq b+1$.
    For any $\boldsymbol{x}\in \Sigma_q^n$, $\boldsymbol{x}_{\mathcal{K}}$ is $b$-run of $\boldsymbol{x}$ if and only if $\boldsymbol{x}_{[1,n]\setminus [k,k+b-1]}= \boldsymbol{x}_{[1,n]\setminus [k',k'+b-1]}$ for any $k,k'\in [k_1,k_2-b+1]$.
\end{lemma}

\begin{lemma}\label{lem:del_ball_t=1}
    For any $n\geq b+1$ with $b\geq 1$, we have $|\mathcal{D}_{1,b}(\boldsymbol{x})|
    = 1+ |\{j\in [b+1,n]: x_j\neq x_{j-b}\}| \leq n-b+1$.
    Moreover, $|\mathcal{D}_{1,b}(\boldsymbol{x})|=n-b+1$ holds if and only if $x_j\neq x_{j-b}$ for any $j\in [b+1,n]$.
\end{lemma}

\begin{IEEEproof}
    By Lemma \ref{lem:del_ball}, we know that $|\mathcal{D}_{1,b}(\boldsymbol{x})|$ counts the number of $b$-runs of $\boldsymbol{x}$, which is given by $1+ |\{j\in [b+1,n]: x_j\neq x_{j-b}\}|$.
    Then the conclusion follows.
\end{IEEEproof}

As a byproduct of Equation (\ref{eq:del_ball}), we can derive the following conclusion more simply than \cite[Lemma 8]{Schoeny-17-IT-BD} and \cite[Claim 4]{Wang-24-IT-BD}. Furthermore, our result is applicable for any $n \geq b$, whereas they only considered the case where $b \mid n$.

\begin{lemma}
    For any $n\geq b+1$ with $b\geq 1$ and $i \in [1,n-b+1]$, we have $|\{\boldsymbol{x}\in \Sigma_q^n: |\mathcal{D}_{1,b}(\boldsymbol{x})|=i \}|= q^b(q-1)^{i-1}\binom{n-b}{i-1}$.
\end{lemma}

\begin{IEEEproof}
    When $|\mathcal{D}_{1,b}(\boldsymbol{x})|=i$, by Equation (\ref{eq:del_ball}), we observe the following
    \begin{itemize}
        \item the first $b$ symbols of $\boldsymbol{x}$ can be chosen arbitrarily;

        \item in the remaining $n-b$ positions, there are exactly $i-1$ choices for $j$ such that $x_j \neq x_{j-b}$;

        \item for each of these $i-1$ positions, there are $q-1$ options for the symbols;

        \item for the other positions, there is exactly one option for the symbols.
    \end{itemize}
    Consequently, the total number of possible selections for $\boldsymbol{x}$ is $q^{b}\binom{n-b}{i-1}(q-1)^{i-1}$, thereby completing the proof.

\end{IEEEproof}

Now we focus on the general parameter $t$ and examine the maximum size of a $b$-burst-deletion ball of radius $t$. When $t = 1$, the $b$-burst-deletion ball reaches its maximum size when the center contains the highest number of $b$-runs. This observation suggests that for $t \geq 2$, the error ball may also be maximized at a center with a similarly high number of $b$-runs.
In particular, for the special case where $b=1$, \cite[Theorem 1]{Hirschberg-99} proved this by taking the cyclic sequence $\boldsymbol{x}=x_1x_2\cdots x_n$, where $x_i\equiv i-1 \pmod{q}$ for $i\in [1,n]$.
To this end, we define the following $b$-cyclic sequence, whose $b$-burst-deletion ball will be shown to achieve the maximum size among all $b$-burst-deletion balls centered at sequences in $\Sigma_q^n$.

\begin{definition}
  For any $\sigma\in \Sigma_q$ and $n\geq 1$, we define $\boldsymbol{X}_{n,q,b}^{\sigma}= X_1 X_2 \cdots X_n\in \Sigma_q^n$ as a \textbf{$b$-cyclic sequence starting with $\sigma^b$} if $X_i\equiv \sigma+\lfloor (i-1)/b \rfloor \pmod{q}$ for $i\in [1,n]$. In other words, let $m= \lceil n/bq \rceil$ and $\boldsymbol{x}= 0^b \circ 1^b \circ \cdots \circ (q-1)^b$, 
  then we have $\boldsymbol{X}_{n,q,b}^{0}= (\underbrace{\boldsymbol{x}\circ \boldsymbol{x}\circ \cdots \circ \boldsymbol{x}}_{m})_{[1,n]}$ and $\boldsymbol{X}_{n,q,b}^{\sigma}= \big( \sigma^{b} \circ (\sigma+1)^{b} \circ \cdots \circ (q-1)^{b} \circ \boldsymbol{X}_{n,q,b}^{0} \big)_{[1,n]}$.
\end{definition}

\begin{example}
  Assume $n=10$, $q=3$, and $b=2$, we have $\boldsymbol{X}_{n,q,b}^{0}= 0011220011$, $\boldsymbol{X}_{n,q,b}^{1}= 1122001122$, and $\boldsymbol{X}_{n,q,b}^{2}= 2200112200$.
\end{example}

A simple but useful observation is that the size of the $b$-burst-deletion ball centered at a $b$-cyclic sequence is independent of the starting symbol and is determined solely by the length of the sequence and the size of the alphabet.

\begin{lemma}
  Assume $q\geq 2$, $t,b\geq 1$ and $n\geq bt+1$. For any $\sigma,\sigma'\in \Sigma_q$, we have $|\mathcal{D}_{t,b}(\boldsymbol{X}_{n,q,b}^{\sigma})|= |\mathcal{D}_{t,b}(\boldsymbol{X}_{n,q,b}^{\sigma'})|$.
\end{lemma}

\begin{IEEEproof}
    Consider the bijection mapping $\psi: \Sigma_q \rightarrow \Sigma_q$ defined by $\psi(\alpha)\equiv \alpha+\sigma-\sigma' \pmod{q}$ for $\alpha \in \Sigma_q$. Assume that $\boldsymbol{x}= x_1x_2\cdots x_{n-bt}$ is obtained from $\boldsymbol{X}_{n,q,b}^{\sigma}$ after $t$ bursts of deletions, each of length exactly $b$.
    We then define $\psi(\boldsymbol{x}):= \psi(x_1) \psi(x_2) \cdots \psi(x_{n-bt})$. 
    It follows that $\psi(\boldsymbol{x})$ can also be obtained from $\boldsymbol{X}_{n,q,b}^{\sigma'}$ after $t$ times of $b$-burst-deletion operations at the same positions from which $\boldsymbol{x}$ was derived from $\boldsymbol{X}_{n,q,b}^{\sigma}$.
    This implies $|\mathcal{D}_{t,b}(\boldsymbol{X}_{n,q,b}^{\sigma})|\leq  |\mathcal{D}_{t,b}(\boldsymbol{X}_{n,q,b}^{\sigma'})|$.
    Similarly, we can also obtain $|\mathcal{D}_{t,b}(\boldsymbol{X}_{n,q,b}^{\sigma})|\geq  |\mathcal{D}_{t,b}(\boldsymbol{X}_{n,q,b}^{\sigma'})|$, thereby completing the proof.
\end{IEEEproof}

As a result, we can define 
\begin{align*}
  d_{q,b}(n,t):= |\mathcal{D}_{t,b}(\boldsymbol{X}_{n,q,b}^{\sigma})|,
\end{align*}
for $\sigma\in \Sigma_q$ and $n\geq bt+1$.
Moreover, we set $d_{q,b}(n,t)= 1$ if $n=bt$ and $d_{q,b}(n,t)= 0$ if $n<bt$.

When $b=1$, the value of $d_{q,b}(n,t)$ has been established in \cite{Hirschberg-99}.

\begin{lemma}\cite[Theorem 2.6]{Hirschberg-99}\label{lem:size_b=1}
    For any $n\geq t+1$ with $t\geq 1$, we have $d_{q,1}(n,t)= \sum_{i=0}^t \binom{n-t}{i} d_{q-1,1}(t,t-i)$.
\end{lemma}

In what follows, we will determine the quantity $d_{q,b}(n,t)$ for $b\geq 2$, which constitutes the main contribution of this subsection. 

\begin{theorem}\label{thm:del_ball}
    For any $n\geq bt+1$ with $t\geq 1$ and $b\geq 1$, we have $d_{q,b}(n,t)=D_{q,b}(n,t)=  \sum_{i=0}^t \binom{n-bt}{i} d_{q-1,1}(t,t-i)$, where $D_{q,b}(n,t)=\max\{|\mathcal{D}_{t,b}(\boldsymbol{x})|: \boldsymbol{x} \in \Sigma_q^n\}$.
    In particular, when $q=2$, we have $d_{2,b}(n,t)=D_{2,b}(n,t)= \sum_{i=0}^t \binom{n-bt}{i}$. 
\end{theorem}

We begin by exploring several properties of $d_{q,b}(n,t)$.

\begin{lemma}\label{lem:del_inequality}
    For any $n\geq bt+1$ with $t,b\geq 1$, we have $d_{q,b}(n-b,t-1) \leq d_{q,b}(n,t)$.
\end{lemma}

\begin{IEEEproof}
    Since $\boldsymbol{X}_{n-b,q,b}^{1} \in \mathcal{D}_{1,b}(\boldsymbol{X}_{n,q,b}^{0})$, by the first statement of Claim \ref{cla:del}, we obtain $\mathcal{D}_{t-1,b}(\boldsymbol{X}_{n-b,q,b}^{1}) \subseteq \mathcal{D}_{t,b}(\boldsymbol{X}_{n,q,b}^{0})$, then the conclusion follows.
\end{IEEEproof}

\begin{lemma}\label{lem:del_rec}
    For any $n\geq bt+1$ with $t,b\geq 1$, we have $d_{q,b}(n,t)= \sum_{i=0}^{q-1}d_{q,b}(n-ib-1,t-i)$.
\end{lemma}

\begin{IEEEproof}
    Let $\sigma_0$ be the last symbol of $\boldsymbol{X}_{n,q,b}^{0}$ and let $\sigma_i\in \Sigma_q$ be defined such that $\sigma_i \equiv \sigma_0-i \pmod{q}$ for $i\in [1,q-1]$. By the definition of $\boldsymbol{X}_{n,q,b}^{0}$, we observe that the last appearance of $\sigma_i$ in $\boldsymbol{X}_{n,q,b}^{0}$ belongs to the interval $[n-ib,n-(i-1)b-1]$.
    Then, by the last statement of Claim \ref{cla:del}, we can obtain
    \begin{align*}
        \mathcal{D}_{t,b}(\boldsymbol{X}_{n,q,b}^{0})
        &= \sqcup_{i=0}^{q-1} \mathcal{D}_{t,b}(\boldsymbol{X}_{n,q,b}^{0})_{\sigma_i} \\
        &= \sqcup_{i=0}^{q-1} \mathcal{D}_{t-i,b}(\boldsymbol{X}_{n-ib-1,q,b}^{0}) \circ \sigma_i.
    \end{align*}
    Then the conclusion follows.
\end{IEEEproof}

Now we are prepared to demonstrate that the maximum size of a $b$-burst-deletion ball of radius $t$ is equal to $D_{q,b}(n,t)$.

\begin{theorem}\label{thm:del_size}
    Let $n\geq bt+1$ with $t,b\geq 1$. For any $\sigma\in \Sigma_q$ and $j \in [0,b-1]$, let 
    \begin{equation}\label{eq:Y}
      \boldsymbol{Y}_{n,q,b}^{\sigma,j}:= \sigma^j \circ \boldsymbol{X}_{n-j,q,b}^{\sigma+1},
    \end{equation}
    where $\sigma+1\in \Sigma_q$, then $d_{q,b}(n,t)= |\mathcal{D}_{t,b}(\boldsymbol{Y}_{n,q,b}^{\sigma,j})|=D_{q,b}(n,t)$, where $D_{q,b}(n,t)=\max\{|\mathcal{D}_{t,b}(\boldsymbol{x})|: \boldsymbol{x} \in \Sigma_q^n\}$.
\end{theorem}

\begin{IEEEproof}
    We prove the theorem by induction on $n+t$.
    For the base case where $n+t=b+2$, we have $n=b+1$ and $t=1$.
    In this case, we observe that $\mathcal{D}_{t,b}(\boldsymbol{x})= \{x_1,x_{b+1}\}$ and that $|\mathcal{D}_{t,b}(\boldsymbol{x})|\leq 2$ for any $\boldsymbol{x} \in \Sigma_q^n$. Since the first and $(b+1)$-th symbols in $\boldsymbol{Y}_{n,q,b}^{\sigma,j}$ are distinct, the conclusion is valid for the base case.
    
    Now, assuming that the conclusion is valid for $n+t<m$, we examine the scenario where $n+t=m\geq b+3$.
    Similar to the proof of Lemma \ref{lem:del_rec}, let $\sigma_0$ be the last symbol of $\boldsymbol{Y}_{n,q,b}^{\sigma,j}$ and let $\sigma_i\in \Sigma_q$ be defined such that $\sigma_i\equiv \sigma_0-i \pmod{q}$ for $i\in [1,q-1]$, we have
    \begin{align*}
        \mathcal{D}_{t,b}(\boldsymbol{Y}_{n,q,b}^{\sigma,j})
        &= \sqcup_{i=0}^{q-1} \mathcal{D}_{t,b}(\boldsymbol{Y}_{n,q,b}^{\sigma,j})_{\sigma_i} \\
        &= \sqcup_{i=0}^{q-1} \mathcal{D}_{t-i,b}(\boldsymbol{Y}_{n-ib-1,q,b}^{\sigma,j}) \circ \sigma_i.
    \end{align*}
    Then we can compute 
    \begin{align*}
       |\mathcal{D}_{t,b}(\boldsymbol{Y}_{n,q,b}^{\sigma,j})|
        &= \sum_{i=0}^{q-1}d_{q,b}(n-ib-1,t-i) \\
        &= d_{q,b}(n,t),
    \end{align*}
    where the first equality follows by induction hypothesis and the second equality follows by Lemma \ref{lem:del_rec}.
    
    To complete the proof, we need to show that $|\mathcal{D}_{t,b}(\boldsymbol{x})|\leq d_{q,b}(n,t)$ for any $\boldsymbol{x} \in \Sigma_q^n$.
    To this end, for any $i\in [0,q-1]$, let $f_i$ denote the smallest index such that $x_{n-f_ib}= i$ if such an index exists; otherwise, set $f_i=\frac{n}{b}$. 
    There exists an arrangement of the elements $(0,1,\ldots,q-1)$, denoted as $(\theta_0,\theta_1,\ldots,\theta_{q-1})$, such that $f_{\theta_0} \leq f_{\theta_1} \leq \cdots \leq f_{\theta_{q-1}}$. 
    Observe that $f_{\theta_i} \geq i$ for $i\in [0,q-1]$, then by the last statement of Claim \ref{cla:del}, we can compute
    \begin{align*}
        |\mathcal{D}_{t,b}(\boldsymbol{x})|
        &= \sum_{i=0}^{q-1} |\mathcal{D}_{t,b}(\boldsymbol{x})_{\theta_i}| \\
        &= \sum_{i=0}^{q-1} |\mathcal{D}_{f_{\theta_i},b}(\boldsymbol{x}_{[1,n-f_{\theta_i}b-1]})|\\
        &\stackrel{(\ast)}{\leq} \sum_{i=0}^{q-1} d_{q,b}(n-f_{\theta_i}b-1,t-f_{\theta_i})\\
        &\stackrel{(\star)}{\leq} \sum_{i=0}^{q-1} d_{q,b}(n-ib-1,t-i)\\
        &\stackrel{(\diamond)}{=} d_{q,b}(n,t),
    \end{align*}
    where $(\ast)$ follows by induction hypothesis, $(\star)$ follows by Lemma \ref{lem:del_inequality}, and $(\diamond)$ follows by Lemma \ref{lem:del_rec}.
    Consequently, the conclusion is also valid for $n+t=m$, thereby completing the proof.
\end{IEEEproof}

Now we are ready to prove the main theorem of this subsection.

\smallskip

\begin{IEEEproof}[Proof of Theorem \ref{thm:del_ball}]
    By Theorem \ref{thm:del_size}, we have $d_{q,b}(n,t)= D_{q,b}(n,t)$. In the following, we prove the rest conclusion by induction on $n+t$. For the base case where $n+t=b+2$, we have $n=b+1$ and $t=1$.
    Then the conclusion follows by Lemma \ref{lem:size_b=1}. Now, assuming that the conclusion is valid for $n+t<m$, we examine the scenario where $n+t= m \geq b+3$.
    By Lemma \ref{lem:del_rec}, we have $D_{q,b}(n,t) = \sum_{i=0}^{q-1} D_{q,b}(n - ib - 1, t - i)$. Below we divide the computation of $D_{q,b}(n,t)$ into two steps. 
    
    In the first step, we compute the term $\sum_{i=1}^{q-1} D_{q,b}(n - ib - 1, t - i)$, which is given by
    \begin{equation}\label{eq:first_step}
    \begin{aligned}
      \sum_{i=1}^{q-1} D_{q,b}(n-ib-1,t-i)
        &\stackrel{(\ast)}{=}\sum_{i=1}^{q-1} \sum_{j=0}^{t-i} \binom{n-bt-1}{j} D_{q-1,1}(t-i,t-i-j) \\
        &\stackrel{(\star)}{=}\sum_{i=1}^{q-1} \sum_{j=0}^{t-1} \binom{n-bt-1}{j} D_{q-1,1}(t-i,t-i-j) \\
        &=\sum_{i=0}^{q-2} \sum_{j=1}^{t} \binom{n-bt-1}{j-1} D_{q-1,1}(t-i-1,t-i-j) \\
        &=\sum_{j=1}^{t} \binom{n-bt-1}{j-1} \sum_{i=0}^{q-2} D_{q-1,1}(t-i-1,t-i-j) \\
        &\stackrel{(\diamond)}{=}\sum_{j=1}^{t} \binom{n-bt-1}{j-1} D_{q-1,1}(t,t-j),
    \end{aligned}
    \end{equation}
    where $(\ast)$ follows by induction hypothesis, $(\star)$ holds since $D_{q-1,1}(r,s)=0$ for $s>r$, and $(\diamond)$ follows by Lemma \ref{lem:del_rec}.
    
    In the second step, we compute $D_{q,b}(n,t)$ as follows:
    \begin{align*}
        D_{q,b}(n,t) 
        &= D_{q,b}(n-1,t)+ \sum_{i=1}^{q-1} D_{q,b}(n-ib-1,t-i) \\
        &\stackrel{(\ast)}{=}\sum_{j=0}^{t} \binom{n-bt-1}{j} D_{q-1,1}(t,t-j)+ \sum_{j=1}^{t} \binom{n-bt-1}{j-1} D_{q-1,1}(t,t-j) \\
        &\stackrel{(\star)}{=}\sum_{j=0}^{t} \binom{n-bt}{j} D_{q-1,1}(t,t-j),
    \end{align*}
    where $(\ast)$ follows by induction hypothesis and Equation (\ref{eq:first_step}), and $(\star)$
    holds since $\binom{n-bt}{0}= \binom{n-bt-1}{0}=1$ and $\binom{n-bt}{i}= \binom{n-bt-1}{i}+\binom{n-bt-1}{i-1}$ for $i\in [1,t]$.
    Consequently, the conclusion is also valid for $n+t=m$, thereby completing the proof.
\end{IEEEproof}

\subsection{The Maximum Intersection Size Between Two Radius-$t$ $b$-Burst-Deletion Balls}\label{subsec:del_ball_int}

In this subsection, we will determine the maximum intersection size between two burst-deletion balls.
Recall that 
\begin{align*}
    N_{q,b}^-(n,t)=\max\{ |\mathcal{D}_{t,b}(\boldsymbol{x}) \cap \mathcal{D}_{t,b}(\boldsymbol{y})|: \boldsymbol{x} \neq \boldsymbol{y} \in \Sigma_q^n\}.
\end{align*}

When $b=1$, the value of $N_{q,b}^-(n,t)$ has been established in \cite{Levenshtein-01-JCTA-recons}.

\begin{lemma}\cite[Equation (28), Lemma 1, and Theorem 1]{Levenshtein-01-JCTA-recons}\label{lem:del_int_b=1}
    For any $t\geq 1$, $q\geq 2$, and $n\geq t+2$, we have $N_{q,1}^-(n,t)= D_{q,1}(n,t)- D_{q,1}(n-1,t)+ D_{q,1}(n-2,t-1)$.
\end{lemma}

One may expect that $N_{q,b}^-(n,t)= D_{q,b}(n,t)- D_{q,b}(n-b,t)+ D_{q,b}(n-2b,t-1)$ holds for any $b\geq 2$.
However, we will demonstrate that this is not the case, at least when $q=2$.
To explain this abnormal phenomenon, we first present the following conclusion derived in \cite{Sun-23-IT-BDR}.

\begin{lemma}\cite[Theorem 3.4]{Sun-23-IT-BDR}\label{lem:del_int}
For any $b\geq 1$, $q\geq 2$, and $n\geq 2b-1$, we have $N_{q,b}^-(n,1)= \max\{2,b\}$.
\end{lemma}

From Lemma \ref{lem:del_int}, we observe a distinction between the cases where $b=1$ and $b\geq 2$.
Intuitively, we expect that when $t\geq 2$, several differences may arise between the cases where $b=1$ and $b\geq 2$, making it challenging to estimate the value of $N_{q,b}^-(n,t)$.
In what follows, we will primarily consider the case where $q=2$ and determine the exact value of $N_{2,b}^-(n,t)$.

\begin{theorem}\label{thm:del_int}
For $n\geq (t+1)b-1$ with $b\geq 2$ and $t\geq 1$, we have 
\begin{align*}
  N_{2,b}^-(n,t)
  &= D_{2,b}(n,t) - D_{2,b}(n-b,t)+ D_{2,b}(n-3b,t-2)\\
  &=D_{2,b}(n,t) - \binom{n-(t+1)b+1}{t},
\end{align*}
where the last equality follows by Theorem \ref{thm:del_ball}.
\end{theorem}

Before proceeding with the proof, we present the following conclusions that will be used frequently.

\begin{lemma}\label{lem:del_inequality'}
  For any $n\geq bt+1$ with $t,b\geq 1$, we have $D_{q,b}(n,t) \leq D_{q,b}(n+1,t)$.
\end{lemma}

\begin{IEEEproof}
    Consider the sequences $\boldsymbol{Y}_{n,q,b}^{0,0}$ and $\boldsymbol{Y}_{n+1,q,b}^{0,1}$ (defined in Equation (\ref{eq:Y})), then by Theorem \ref{thm:del_size}, we obtain $\mathcal{D}_{t,b}(\boldsymbol{Y}_{n,q,b}^{0,0})= D_{q,b}(n,t)$ and $\mathcal{D}_{t,b}(\boldsymbol{Y}_{n+1,q,b}^{0,1})= D_{q,b}(n+1,t)$.
    Moreover, by the second statement of Claim \ref{cla:del}, we have $\mathcal{D}_{t,b}(\boldsymbol{Y}_{n+1,q,b}^{0,1})^0= 0\circ \mathcal{D}_{t,b}(\boldsymbol{Y}_{n,q,b}^{0,0})$, which implies that $D_{q,b}(n,t) \leq D_{q,b}(n+1,t)$.
\end{IEEEproof}

\begin{lemma}\label{lem:prefix}
  Assume $n\geq bt+1$ with $b,t\geq 1$. For any $\boldsymbol{x}\in \Sigma_q^n$ and $\boldsymbol{z}\in \Sigma_q^{j}$ with $j\in [1,n]$, if there exists some $f\geq 0$ such that $\boldsymbol{x}_{[1,fb+j]}$ is the shortest prefix of $\boldsymbol{x}$ that can yield $\boldsymbol{z}$ through $b$-burst-deletion operations, then $|\mathcal{D}_{t,b}(\boldsymbol{x})^{\boldsymbol{z}}|\leq D_{q,b}(n-fb-j,n-f)$.
  Moreover, if $f\geq s$ for some integer $s$, we have $|\mathcal{D}_{t,b}(\boldsymbol{x})^{\boldsymbol{z}}|\leq D_{q,b}(n-sb-j,n-s)$.
\end{lemma}

\begin{IEEEproof}
    Similar to the second statement of Claim \ref{cla:del}, we observe that $\mathcal{D}_{t,b}(\boldsymbol{x})^{\boldsymbol{z}}= \boldsymbol{z}\circ \mathcal{D}_{t-f,b}(\boldsymbol{x}_{[fb+j+1,n]})$. 
    Then we can compute 
    \begin{align*}
      |\mathcal{D}_{t,b}(\boldsymbol{x})^{\boldsymbol{z}}|
      &= |\mathcal{D}_{t-f,b}(\boldsymbol{x}_{[fb+j+1,n]})|\\
      &\stackrel{(\ast)}{\leq} D_{2,b}(n-fb-1,t-f)\\
      &\stackrel{(\star)}{\leq} D_{2,b}(n-sb-1,t-s),
    \end{align*}
    where $(\ast)$ follows by Theorem \ref{thm:del_size} and $(\star)$ follows by Lemma \ref{lem:del_inequality}.
\end{IEEEproof}

We first show that $N_{2,b}^-(n,t)\geq D_{2,b}(n,t) - D_{2,b}(n-b,t)+ D_{2,b}(n-3b,t-2)$ by finding two sequences such that the intersection size between their $b$-burst-deletion balls of radius $t$ equals $D_{2,b}(n,t) - D_{2,b}(n-b,t)+ D_{2,b}(n-3b,t-2)$.

\begin{lemma}\label{lem:del_lb}
    For any $q\geq 2$ and $n\geq (t+1)b-1$ with $b\geq 2$ and $t\geq 1$, let $\boldsymbol{x}=0^b\circ \boldsymbol{X}_{n-b,q,b}^{1}$ and
    and $\boldsymbol{y}= 0^{b-1}1 \circ \boldsymbol{X}_{n-b,q,b}^{1}$.
    In other words, the sequence $\boldsymbol{x}$ is a $b$-cyclic sequence that starts with $0^b$, and the sequence $\boldsymbol{y}$ is derived from $\boldsymbol{x}$ by flipping its $b$-th entry.
    Then we have 
    \begin{align*}
        |\mathcal{D}_{t,b}(\boldsymbol{x})\cap \mathcal{D}_{t,b}(\boldsymbol{y})|= D_{q,b}(n,t) - D_{q,b}(n-b,t)+ D_{q,b}(n-(q+1)b,t-q).
    \end{align*}
\end{lemma}

\begin{IEEEproof}
    Let $\mathcal{S}:= \mathcal{D}_{t,b}(\boldsymbol{x})\cap \mathcal{D}_{t,b}(\boldsymbol{y})$, we can partition $\mathcal{S}$ into $(q-1)b+1$ disjoint subsets:
    \begin{align*}
        \mathcal{S}= \mathcal{S}^{0^b} \sqcup \big(\sqcup_{i=0}^{b-1} \sqcup_{j=1}^{q-1} \mathcal{S}^{0^{i}j} \big).
    \end{align*}
    To determine the size of $\mathcal{S}$, it suffices to calculate the size of these $(q-1)b+1$ disjoint subsets separately.
    \begin{itemize}
        \item \textbf{The size of $\mathcal{S}^{0^b}$:} Observe that $0^{b-1}\circ 1 \circ 1^b \circ \cdots \circ (q-1)^b \circ 0^b$ represents the shortest prefix of $\boldsymbol{y}$ that can produce $0^b$ through $b$-burst-deletion operations.
            Then by Lemma \ref{lem:prefix}, we get 
        \begin{align*}
            \mathcal{D}_{t,b}(\boldsymbol{y})^{0^b}= 0^b \circ \mathcal{D}_{t-q,b}(\boldsymbol{X}_{n-(q+1)b,q,b}^{1}).
        \end{align*}
        Moreover, we have 
        \begin{align*}
            \mathcal{D}_{t,b}(\boldsymbol{x})^{0^b}= 0^b \circ \mathcal{D}_{t,b}(\boldsymbol{X}_{n-b,q,b}^{1}).
        \end{align*}
        Then we can compute
        \begin{align*}
        \mathcal{S}^{0^b}
        &= 0^b \circ \big( \mathcal{D}_{t,b}(\boldsymbol{X}_{n-b,q,b}^{1}) \cap \mathcal{D}_{t-q,b}(\boldsymbol{X}_{n-(q+1)b,q,b}^{1}) \big) \\
        &= 0^b \circ \mathcal{D}_{t-q,b}(\boldsymbol{X}_{n-(q+1)b,q,b}^{1}),
        \end{align*}
        where the last equality follows by the first statement of Claim \ref{cla:del} and the fact that $\boldsymbol{X}_{n-(q+1)b,q,b}^{1} \in \mathcal{D}_{q,b}(\boldsymbol{X}_{n-b,q,b}^{1})$.
        Finally, by Theorem \ref{thm:del_size}, we derive
        \begin{align*}
            |\mathcal{S}^{0^b}|
            &= |\mathcal{D}_{t-q,b}(\boldsymbol{X}_{n-(q+1)b,q,b}^{1})|\\
            &=D_{q,b}(n-(q+1)b,t-q).
        \end{align*}

        \item \textbf{The size of $\mathcal{S}^{j}$ for $j\in [1,q-1]$:} Observe that $0^b\circ 1^b \circ \cdots \circ (j-1)^b \circ j$ and $0^{b-1} \circ 1 \circ 1^b \circ \cdots \circ (j-1)^b \circ j$ are the shortest prefixes of $\boldsymbol{x}$ and $\boldsymbol{y}$, respectively, that can yield $j$ through $b$-burst-deletion operations.
            Then by Lemma \ref{lem:prefix}, we get
        \begin{align*}
        \mathcal{D}_{t,b}(\boldsymbol{x})^{j}
        =\mathcal{D}_{t,b}(\boldsymbol{y})^{j} = j \circ \mathcal{D}_{t-j,b}(\boldsymbol{Y}_{n-jb-1,q,b}^{j,b-1}),
        \end{align*}
        where $\boldsymbol{Y}_{n,q,b}^{j,b-1}$ is defined in Equation (\ref{eq:Y}).
        It follows by Theorem \ref{thm:del_size} that 
        \begin{align*}
            |\mathcal{S}^{j}|
            &= |\mathcal{D}_{t-j,b}(\boldsymbol{Y}_{n-jb,q,b}^{j,b-1})|\\
            &=D_{q,b}(n-jb-1,t-j).
        \end{align*}

        \item \textbf{The size of $\mathcal{S}^{0^i j}$ for $i\in [1,b-1]$ and $j\in [1,q-1]$:} Observe that $0^b \circ 1^b \circ \cdots \circ (j-1)^b \circ j^{i+1}$ represents the shortest prefixes of $\boldsymbol{x}$ that can yield $0^i j$ through $b$-burst-deletion operations.
            Then by Lemma \ref{lem:prefix}, we get
        \begin{align*}
        \mathcal{D}_{t,b}(\boldsymbol{x})^{0^i j}
        =0^i j \circ \mathcal{D}_{t-j,b}(\boldsymbol{Y}_{n-jb-i-1,q,b}^{j,b-i-1}),
        \end{align*}
        where $\boldsymbol{Y}_{n,q,b}^{j,b-i-1}$ is defined in Equation (\ref{eq:Y}).
        By a similar argument, but noting the subtle difference that $\boldsymbol{y}$ begins with $0^{b-1}1$, we obtain 
        \begin{equation*}
            \mathcal{D}_{t,b}(\boldsymbol{y})^{0^i j}= 
            \begin{cases}
                0^{b-1}1 \circ \mathcal{D}_{t,b}(\boldsymbol{X}_{n-b,q,b}^{0,0}), &\mbox{if } i=b-1, j=1\\
                0^i j \circ \mathcal{D}_{t-j,b}(\boldsymbol{Y}_{n-jb-i-1,q,b}^{j,b-i-1}), &\mbox{otherwise}.
            \end{cases}
        \end{equation*}
        Since by definition, we have $\boldsymbol{Y}_{n-2b,q,b}^{1,0}\in \mathcal{D}_{1,b}(\boldsymbol{Y}_{n-b,q,b}^{0,0})$.
        Then we derive
        \begin{align*}
            \mathcal{S}^{0^i j}
            &= 0^i j \circ \mathcal{D}_{t-j,b}(\boldsymbol{Y}_{n-jb-i-1,q,b}^{j,b-i-1}).
        \end{align*}
        It follows by Theorem \ref{thm:del_size} that
        \begin{align*}
            |\mathcal{S}^{0^i j}|
            &= |\mathcal{D}_{t-j,b}(\boldsymbol{Y}_{n-jb-i-1,q,b}^{1,b-i-1})|
            = D_{q,b}(n-i-jb-1,t-j).
        \end{align*}
    \end{itemize}
    As a result, we can compute
    \begin{equation}\label{eq:del_rec}
    \begin{aligned}
      |\mathcal{S}|
      &= |\mathcal{S}^{0^b}|+ \sum_{i=0}^{b-1}\sum_{j=1}^{q-1}|\mathcal{S}^{0^ij}|\\
      &= D_{q,b}(n-(q+1)b,t-q)+ \sum_{i=0}^{b-1} \sum_{j=1}^{q-1}D_{q,b}(n-i-jb-1,t-j) \\
      &\stackrel{(\ast)}{=} D_{q,b}(n-(q+1)b,t-q)+ \sum_{i=0}^{b-1}[D_{q,b}(n-i,t)-D_{q,b}(n-i-1,t)]\\
      &=D_{q,b}(n-(q+1)b,t-q)+ D_{q,b}(n,t) - D_{q,b}(n-b,t),
    \end{aligned}
    \end{equation}
    where $(\ast)$ follows by Lemma \ref{lem:del_rec}. 
    This completes the proof.
\end{IEEEproof}  

\smallskip

\begin{IEEEproof}[Proof of Theorem \ref{thm:del_int}]
    Based on Lemma \ref{lem:del_lb}, it remains to show that $N_{2,b}^-(n,t)\leq D_{2,b}(n,t) - D_{2,b}(n-b,t)+ D_{2,b}(n-3b,t-2)$. Now we prove it by induction on $n+t$.
    
    For the base case where $n+t=2b$, we have $n=2b-1$ and $t=1$.
    Observe that $n-3b<n-b<b$, we have $D_{2,b}(n-3b,t-2)= D_{2,b}(n-b,t)= 0$.
    By Lemma \ref{lem:del_ball_t=1}, we know that $D_{2,b}(n,t)=b$. Therefore, we can compute $D_{2,b}(n,t) - D_{2,b}(n-b,t)+ D_{2,b}(n-3b,t-2)=b$.
    Furthermore, by Lemma \ref{lem:del_int}, we find that $N_{2,b}^-(n,t)=b$, which confirms that the conclusion is valid for the base case.
    
    Now, assuming the conclusion is valid for $n+t<m$, we examine the scenario where $n+t= m \geq 2b+1$.
    For any two distinct sequences $\boldsymbol{x}, \boldsymbol{y}\in \Sigma_2^n$, let $\mathcal{S}:= \mathcal{D}_{t,b}(\boldsymbol{x})\cap \mathcal{D}_{t,b}(\boldsymbol{y})$, we distinguish between the following four cases.
    In the following discussion, for the sake of simplicity in notation, let $\overline{\alpha}:= 1-\alpha$ for $\alpha\in \Sigma_2$.
    \begin{itemize}
        \item \textbf{The case of $x_1\neq y_1$:} Let $f$ denote the smallest index such that $x_{fb+1}= y_1$ if such an index exists; otherwise, set $f=\frac{n}{b}$. 
            Clearly, we have $f \geq 1$. Then by Lemma \ref{lem:prefix}, we have
            \begin{equation*}
              |\mathcal{D}_{t,b}(\boldsymbol{x})^{y_1}|\leq D_{2,b}(n-b-1,t-1).
            \end{equation*}
            Similarly, we can obtain
            \begin{align*}
              |\mathcal{D}_{t,b}(\boldsymbol{y})^{x_1}|
              &\leq D_{2,b}(n-b-1,t-1).
            \end{align*}
        Observe that $\mathcal{S}^{x_1} \subseteq \mathcal{D}_{t,b}(\boldsymbol{y})^{x_1}$ and $\mathcal{S}^{y_1} \subseteq \mathcal{D}_{t,b}(\boldsymbol{x})^{y_1}$, we can compute
        \begin{equation}\label{eq:del_con}
        \begin{aligned}
            |\mathcal{S}|
            &= |\mathcal{S}^{x_1}|+ |\mathcal{S}^{y_1}| \\
            &\leq 2D_{2,b}(n-b-1,t-1) \\
            &\stackrel{(\ast)}{=} D_{2,b}(n,t) - D_{2,b}(n-1,t) + D_{2,b}(n-b-1,t-1) \\
            &\stackrel{(\star)}{\leq} D_{2,b}(n,t) - D_{2,b}(n-b,t) + D_{2,b}(n-b-1,t-1),
        \end{aligned}
        \end{equation}
        where $(\ast)$ follows by Lemma \ref{lem:del_rec} and $(\star)$ follows by Lemma \ref{lem:del_inequality'}.

        \item \textbf{The case of $x_1=y_1$ and we have either $x_{b+1}=x_1$ or $y_{b+1}=y_1$:} Due to symmetry, without loss of generality assume $x_{b+1}=x_1$. In this case, we have $\boldsymbol{x}_{[2,n]}\neq \boldsymbol{y}_{[2,n]}$ and $\mathcal{S}^{x_1}= x_1 \circ \big( \mathcal{D}_{t,b}(\boldsymbol{x}_{[2,n]}) \cap \mathcal{D}_{t,b}(\boldsymbol{y}_{[2,n]}) \big)$.
            By the hypothesis induction, we can obtain 
        \begin{align*}
          |\mathcal{S}^{x_1}|
          &= \big|\mathcal{D}_{t,b}(\boldsymbol{x}_{[2,n]}) \cap \mathcal{D}_{t,b}(\boldsymbol{y}_{[2,n]})\big| \\
           &\leq N_{2,b}^-(n-1,t) \\
           &\leq D_{2,b}(n-1,t) - \binom{n-(t+1)b}{t}.
        \end{align*}
        Moreover, let $f$ denote the smallest index such that $x_{fb+1}= \overline{x_1}$ if such an index exists; otherwise, set $f=\frac{n}{b}$, we have $f \geq 2$. 
        Then by Lemma \ref{lem:prefix}, we get
        \begin{align*}
          |\mathcal{D}_{t,b}(\boldsymbol{x})^{\overline{x_1}}|\leq D_{2,b}(n-2b-1,t-2), 
        \end{align*}
        implying that $|\mathcal{S}^{\overline{x_1}}|\leq D_{2,b}(n-2b-1,t-2)$.
        Then we may compute
        \begin{align*}
          |\mathcal{S}|
          &= |\mathcal{S}^{x_1}|+ |\mathcal{S}^{\overline{x_1}}|\\
          &\leq D_{2,b}(n-1,t) - \binom{n-(t+1)b}{t} + D_{2,b}(n-2b-1,t-2)\\
          &\stackrel{(\ast)}{=} D_{2,b}(n-1,t) - \binom{n-(t+1)b}{t} + \sum_{i=0}^{t-1}\binom{n-bt-1}{i} - \binom{n-bt-1}{t-1}\\
          &\stackrel{(\ast)}{=} D_{2,b}(n-1,t) - \binom{n-(t+1)b}{t} + D_{2,b}(n-b-1,t-1) - \binom{n-bt-1}{t-1}\\
          &\stackrel{(\star)}{\leq} D_{2,b}(n-1,t) - \binom{n-(t+1)b}{t} + D_{2,b}(n-b-1,t-1) - \binom{n-(t+1)b}{t-1}\\
          &\stackrel{\diamond}{=} D_{2,b}(n,t)-\binom{n-(t+1)b+1}{t},
        \end{align*}
        where $(\ast)$ follows by Theorem \ref{thm:del_ball}, $(\star)$ follows by the fact that the function $\binom{n}{t}$ increases as $n$ increases, and $(\diamond)$ follows by Lemma \ref{lem:del_rec}.
        
        
        \item \textbf{The case of $x_1= y_1 \neq x_{b+1}=y_{b+1}$ and $\boldsymbol{x}_{[b+2,n]}\neq \boldsymbol{y}_{[b+2,n]}$:} In this case, by the second statement of Claim \ref{cla:del}, we get
            \begin{align*}
                \mathcal{S}^{x_1}= x_1 \circ \big( \mathcal{D}_{t,b}(\boldsymbol{x}_{[2,n]}) \cap \mathcal{D}_{t,b}(\boldsymbol{y}_{[2,n]}) \big)
            \end{align*}
            and 
            \begin{align*}
                \mathcal{S}^{\overline{x_1}}= \overline{x_1} \circ \big( \mathcal{D}_{t-1,b}(\boldsymbol{x}_{[b+2,n]}) \cap \mathcal{D}_{t-1,b}(\boldsymbol{y}_{[b+2,n]}) \big).
            \end{align*}
            Then by the induction hypothesis, we can compute
            \begin{equation}\label{eq:del_int_rec}
            \begin{aligned}
                |\mathcal{S}|
                &= |\mathcal{S}^{x_1}|+ |\mathcal{S}^{\overline{x_1}}|\\
                &\leq N_{2,b}^-(n-1,t) + N_{2,b}^-(n-b-1,t-1)\\
                &\stackrel{(\ast)}{=} \sum_{i=0}^t \binom{n-bt-1}{i} - \binom{n-(t+1)b}{t} + \sum_{i=0}^{t-1} \binom{n-bt-1}{i} - \binom{n-(t+1)b}{t-1}\\
                &= \binom{n-bt-1}{0}+ \sum_{i=1}^{t} \left[\binom{n-bt-1}{i} + \binom{n-bt-1}{i-1}\right] - \left[\binom{n-(t+1)b}{t} + \binom{n-(t+1)b}{t-1}\right]\\
                &\stackrel{(\star)}{=} \sum_{i=0}^{t} \binom{n-bt}{i} - \binom{n-(t+1)b+1}{t},
            \end{aligned}
            \end{equation}
            where $(\ast)$ follows by Theorem \ref{thm:del_ball} and $(\star)$ holds since $\binom{n-bt}{0}= \binom{n-bt-1}{0}$ and $\binom{n-bt}{i}= \binom{n-bt-1}{i}+\binom{n-bt-1}{i-1}$ for $i\in [1,t]$.
        
         \item \textbf{The case of $x_1= y_1 \neq x_{b+1}=y_{b+1}$ and $\boldsymbol{x}_{[b+2,n]}= \boldsymbol{y}_{[b+2,n]}$:} 
             In this case, the discussion is similar to that of Lemma \ref{lem:del_lb}.
             Let $j$ be the smallest index on which $\boldsymbol{x}$ and $\boldsymbol{y}$ differ, then $j\in [2,b]$.
             We partition $\mathcal{S}$ into the following $j$ disjoint subsets: 
             \begin{align*}
                \mathcal{S}^{\overline{x_1}}, \mathcal{S}^{x_1\overline{x_2}}, \ldots, \mathcal{S}^{x_1\cdots x_{j-2}\overline{x_{j-1}}}, \mathcal{S}^{x_1\cdots x_{j-1}}.
             \end{align*}
             For $i \in [0,j-2]$, let $f_i$ denote the smallest index such that $x_{f_ib+i+1}= \overline{x_{i+1}}$ if such an index exists; otherwise, set $f_i=\frac{n}{b}$.
              Clearly, we have $f_i \geq 1$. 
              Observe that $\boldsymbol{x}_{[1,f_ib+i+1]}$ represents the shortest prefix of $\boldsymbol{x}$ that can produce $x_1\cdots x_{i}\overline{x_{i+1}}$ through $b$-burst-deletion operations.
              Then by Lemma \ref{lem:prefix}, we can compute
             \begin{align*}
               |\mathcal{S}^{x_1\cdots x_{i}\overline{x_{i+1}}}|
               &\leq \mathcal{D}_{t,b}(\boldsymbol{x})^{x_1\cdots x_{i}\overline{x_{i+1}}}\\
               &\leq D_{2,b}(n-b-i-1,t-1).
             \end{align*}
             Now we consider the set $\mathcal{S}^{x_1\cdots x_{j-1}}$ and observe the following:
             \begin{align*}
               \mathcal{S}^{x_1\cdots x_{j-1}}
               &= \mathcal{S}^{x_1\cdots x_{j}} \cup \mathcal{S}^{x_1\cdots x_{j-1} \overline{x_j}}\\
               &\subseteq \mathcal{D}_{t,b}(\boldsymbol{y})^{x_1\cdots x_{j}} \cup \mathcal{D}_{t,b}(\boldsymbol{x})^{x_1\cdots x_{j-1} \overline{x_j}}.
             \end{align*}
             This implies that 
             \begin{align*}
               |\mathcal{S}^{x_1\cdots x_{j-1}}|
                \leq |\mathcal{D}_{t,b}(\boldsymbol{y})^{x_1\cdots x_{j}}|+ |\mathcal{D}_{t,b}(\boldsymbol{x})^{x_1\cdots x_{j-1} \overline{x_j}}|.
             \end{align*}
             Let $f$ denote the smallest index such that $x_{fb+j}= \overline{x_j}$ if such an index exists; otherwise, set $f=\frac{n}{b}$.
             Moreover, let $g$ denote the smallest index such that $y_{gb+j}= x_j$ if such an index exists; otherwise, set $g=\frac{n}{b}$.
             Observe that $x_1\cdots x_{fb+j}$ represents the shortest prefixes of $\boldsymbol{x}$ that can yield $x_1\cdots x_{j-1} \overline{x_j}$ through $b$-burst-deletion operations and $y_1\cdots y_{gb+j}$ represents the shortest prefixes of $\boldsymbol{y}$ that can yield $x_1\cdots x_{j}$ through $b$-burst-deletion operations.
             We distinguish between the following two cases based on whether $x_j=x_{j+b}$.
             \begin{itemize}
               \item If $x_j\neq x_{j+b} = y_{j+b}$, we have $f=1$ and $g\geq 2$.
               Then by Lemma \ref{lem:prefix}, we can compute
                \begin{gather*}
                  |\mathcal{D}_{t,b}(\boldsymbol{y})^{x_1\cdots x_{j}}|\leq D_{2,b}(n-2b-j,t-2),\\
                  |\mathcal{D}_{t,b}(\boldsymbol{x})^{x_1\cdots x_{j-1} \overline{x_j}}|
                  \leq D_{2,b}(n-b-j,t-1).
                \end{gather*}
               \item If $x_j= x_{j+b}=y_{j+b}$, we have $f\geq 2$ and $g=1$.
               Then by Lemma \ref{lem:prefix}, we can compute
               \begin{gather*}
                  |\mathcal{D}_{t,b}(\boldsymbol{y})^{x_1\cdots x_{j}}|
                  \leq  D_{2,b}(n-b-j,t-1),\\
                  |\mathcal{D}_{t,b}(\boldsymbol{x})^{x_1\cdots x_{j-1} \overline{x_j}}|
                  \leq D_{2,b}(n-2b-j,t-2).
                \end{gather*}
             \end{itemize}
             In both cases, we can derive
             \begin{align*}
               |\mathcal{S}^{x_1\cdots x_{j-1}}|
               &\leq |\mathcal{D}_{t,b}(\boldsymbol{y})^{x_1\cdots x_{j}}|+ |\mathcal{D}_{t,b}(\boldsymbol{x})^{x_1\cdots x_{j-1} \overline{x_j}}| \\
               &\leq D_{2,b}(n-b-j,t-1)+ D_{2,b}(n-2b-j,t-2).
             \end{align*}
             Then we can compute
             \begin{align*}
               |\mathcal{S}|
               &= \sum_{i=0}^{j-2} |\mathcal{S}^{x_1\cdots x_{i}\overline{x_{i+1}}}|+ |\mathcal{S}^{x_1\cdots x_{j-1}}| \\
               &\leq \sum_{i=0}^{j-1}D_{2,b} (n-b-i-1,t-1)+ D_{2,b}(n-2b-j,t-2)\\
               &\stackrel{(\ast)}{=} \sum_{i=0}^{j-1} [D_{2,b}(n-i,b)- D_{2,b}(n-i-1,b)]+ D_{2,b}(n-2b-j,t-2)\\
               &= D_{2,b}(n,t)- D_{2,b}(n-j,t)+ D_{2,b}(n-2b-j,t-2)\\
               &\stackrel{(\star)}{\leq} D_{2,b}(n,t)- D_{2,b}(n-b,t)+ D_{2,b}(n-3b,t-2),
             \end{align*}
             where $(\star)$ follows by Lemma \ref{lem:del_rec} and $(\star)$ holds since $j\in [2,b]$ and
             \begin{align*}
               D_{2,b}(n-j,t)- D_{2,b}(n-2b-j,t-2)
               &= \sum_{i=0}^{t}\binom{n-bt-j}{i}- \sum_{i=0}^{t-2}\binom{n-bt-j}{i}\\
               &= \sum_{i=t-1}^{t}\binom{n-bt-j}{i}
             \end{align*}
             is a monotonic non-increasing function with respect to $j$.
    \end{itemize}
    Consequently, the conclusion is also valid for $n+t=m$, thereby completing the proof.
\end{IEEEproof}

By Equation (\ref{eq:del_con}), we have the following conclusion.

\begin{corollary}\label{cor:del_con}
  For any $n\geq bt+1$ with $t\geq 1$ and $b\geq 2$, we have $N_{2,b}^-(n,t)\geq 2D_{2,b}(n-b-1,t-1)$.
\end{corollary}

Moreover, by Equation (\ref{eq:del_int_rec}), the following holds.
\begin{corollary}\label{cor:del_int_rec}
    For any $n\geq bt+1$ with $t\geq 1$ and $b\geq 2$, we have $N_{2,b}^-(n,t) = N_{2,b}^-(n-1,t) + N_{2,b}^-(n-b-1,t-1)$.
\end{corollary}

\subsection{The Reconstruction Algorithm for Burst-Deletion Errors} \label{subsec:del_alg}

We will now design a straightforward algorithm, presented in Algorithm \ref{alg:del}, to recover an arbitrary $\boldsymbol{x} \in \Sigma_2^n$, where $n \geq (t+1)b-1$, given its length $n$ and $N \geq N_{2,b}^-(n, t) + 1$ distinct sequences $\boldsymbol{y}_1, \boldsymbol{y}_2, \ldots, \boldsymbol{y}_N$ in $\mathcal{D}_{t,b}(\boldsymbol{x})$.
Unlike the reconstruction algorithm for bursts of insertions, here we first employ the threshold reconstruction algorithm to recover $ \boldsymbol{x} $, with the exception of at most $ t(b-1) $ positions where the symbols remain undetermined. As a result, we derive a candidate set containing at most $ 2^{t(b-1)} $ sequences, which includes the original sequence $ \boldsymbol{x} $. Since the original sequence is the only one whose error ball encompasses these received sequences, we can correctly recover $ \boldsymbol{x} $ by verifying this property for each candidate in the set. 

The algorithm is operated as follows:
\begin{itemize}
  \item Firstly, we define the set $\mathcal{U} = \{ \boldsymbol{y}_1, \boldsymbol{y}_2, \ldots, \boldsymbol{y}_N \}$, which can be considered as a matrix of size $(n- bt) \times N$, with the vector in the $i$-th column set as $\boldsymbol{y}_i$ for $i \in [1, N]$. 
  \item We then use Lemma \ref{lem:del_con} to recover $x_1$ from the first row of this matrix and identify a submatrix $\mathcal{U}'$ of the matrix $\mathcal{U}$, which has size $(n-bt-1) \times N'$ and consists of distinct columns $\boldsymbol{y}' \in \mathcal{D}_{t',b}(\boldsymbol{x}')$, where $\boldsymbol{x}' = x_{jb+2} \cdots x_n$, $n' = n -bj- 1$, $t' = t - j$, $j\in \{0,1\}$, and $N' \geq N_{2,b}^-(n', t') + 1$.
      It is worth noting that when $j=1$, we can obtain $x_{b+1}=\overline{x_1}$ and the entries $x_2, \ldots, x_b$ are undetermined, which we denote as $?$. 
  \item If $t' = t - j = 0$, then $N_{q,b}^-(n', t')=0$, implying that the matrix $\mathcal{U}'$ consists of only the column $\boldsymbol{x}'$, and recovering $\boldsymbol{x} = x_1\cdots x_{jb+1} \circ \boldsymbol{x}'$ is completed, except for several undetermined entries.
  \item In the case where $t' \geq 1$, recovering $\boldsymbol{x} \in \Sigma_q^n$ reduces to recovering $\boldsymbol{x}' \in \Sigma_q^{n'}$, except for several undetermined entries, with the help of $N' \geq N_{q,b}^-(n', t') + 1$ distinct sequences of length $n'-bt'$ in $\mathcal{D}_{t',b}(\boldsymbol{x}')$. 
  \item Since $n'< n$ and $t' \leq t$, we will eventually reach a stage where $n = 1$ or $t = 0$. At this stage, there are at most $t(b-1)$ undetermined entries.
       Then we can derive a candidate set of size at most $2^{t(b-1)}$ containing $\boldsymbol{x}$ by assigning $0$ and $1$ to those undetermined entries.
  \item Finally, we check all possible sequences in this candidate set and output the one whose $b$-burst-deletion ball of radius $t$ can contain the set $\{\boldsymbol{y}_1,\boldsymbol{y}_2,\ldots,\boldsymbol{y}_N\}$.
\end{itemize}

\begin{algorithm}
    \caption{Sequence Reconstruction Algorithm under Burst-Deletion Channel}\label{alg:del}
    \KwIn{$\mathcal{U}= \{\boldsymbol{y}_1,\boldsymbol{y}_2,\ldots,\boldsymbol{y}_N\} \subseteq \mathcal{D}_{t,b}(\boldsymbol{x})$ for some $\boldsymbol{x}\in \Sigma_2^n$ with $|\mathcal{U}|\geq N_{2,b}^-(n,t)+1$}
    \KwOut{$\boldsymbol{x}= x_1\cdots x_n$}

    \textbf{Initialization:} Set $i=0$ and $\mathcal{T}=\emptyset$\\
    \While{$i<n$ or $t>0$}
        {$i \leftarrow i+1$\\
         Let $\beta \in \Sigma_2$ be such that $|\mathcal{U}^{\beta}|> |\mathcal{U}^{\overline{\beta}}|$\\
         $x_i \leftarrow \beta$\\
         \If{$|\mathcal{U}^{x_i}|> N_{2,b}^-(n-1,t)$}
            {$\mathcal{U}\leftarrow $ the set of sequences that can be obtained from $\mathcal{U}^{x_i}$ by deleting the first entry}
         \If{$|\mathcal{U}^{x_i}|\leq N_{2,b}^-(n-1,t)$}
            {$\mathcal{T}\leftarrow \mathcal{T}\cup [i+1,i+b-1]$\\
             $x_{i+j}\leftarrow ?$ for $j\in [1,b-1]$\\
             $x_{i+b}\leftarrow \overline{x_i}$\\
             $i\leftarrow i+b$\\
             $t\leftarrow t-1$\\
             $\mathcal{U}\leftarrow$ the set of sequences that can be obtained from $\mathcal{U}^{\overline{x_i}}$ by deleting the first entry (by Corollary \ref{cor:del_int_rec}, we have $|\mathcal{U}^{\overline{x_i}}|\geq N_{2,b}^-(n-b-1,t-1)+1$)\\
             \If{$t=0$}
                {$\{x_{i+1} \cdots x_n\}\leftarrow \mathcal{U}$}
             }
        }
    \For{$\boldsymbol{v}\in \Sigma_2^n$ be such that $v_j=x_j$ for $j\notin \mathcal{T}$}
        {\If{$\{\boldsymbol{y}_1,\boldsymbol{y}_2,\ldots,\boldsymbol{y}_N\} \subseteq \mathcal{D}_{t,b}(\boldsymbol{v})$}
            {$\boldsymbol{x}\leftarrow \boldsymbol{v}$}
        }
\end{algorithm}

\begin{lemma}\label{lem:del_con}
    Assume $n\geq b(t+1)-1$ with $t\geq 1$ and $b\geq 2$, for any $\boldsymbol{x}\in \Sigma_2^n$ and $\mathcal{U}\subseteq \mathcal{D}_{t,b}(\boldsymbol{x})$ with $|\mathcal{U}|\geq N_{2,b}^-(n,t) + 1$, we have
    \begin{equation*}
        |\mathcal{U}^{x_1}| > |\mathcal{U}^{\overline{x_1}}|.
    \end{equation*}
    Moreover, if $|\mathcal{U}^{x_1}|\leq N_{2,b}^-(n-1,t)$, then $x_{b+1}= \overline{x_1}$.
\end{lemma}
\begin{IEEEproof}
    Let $f$ denote the smallest index such that $x_{fb+1}= \overline{x_1}$ if such an index exists; otherwise, set $f=\frac{n}{b}$.
    Clearly, we have $f\geq 1$.
    Then by Lemma \ref{lem:prefix}, we can compute
    \begin{align*}
        |\mathcal{U}^{\overline{x_1}}|
        &\leq |\mathcal{D}_{t,b}(\boldsymbol{x})^{\overline{x_1}}|\leq D_{2,b}(n-b-1,t-1).
    \end{align*}
    Since $\mathcal{U}= \mathcal{U}^{x_1} \sqcup \mathcal{U}^{\overline{x_1}}$, by Corollary \ref{cor:del_con}, we have
    \begin{align*}
        |\mathcal{U}^{x_1}|= |\mathcal{U}|- |\mathcal{U}^{\overline{x_1}}|> D_{2,b}(n-b-1,t-1)\geq |\mathcal{U}^{\overline{x_1}}|.
    \end{align*}
    
    If $x_{b+1}=x_1$, we have $f\geq 2$.
    Again by Lemma \ref{lem:prefix}, we can compute
    \begin{align*}
      |\mathcal{U}^{\overline{x_1}}|
        &\leq |\mathcal{D}_{t,b}(\boldsymbol{x})^{\overline{x_1}}|\leq D_{2,b}(n-2b-1,t-2).
    \end{align*}
    Then we have
    \begin{align*}
      N_{2,b}^-(n-b-1,t-1)
      &\stackrel{(\ast)}{\geq} 2D_{2,b}(n-2b-2,t-2)\\
      &\stackrel{(\star)}{\geq} D_{2,b}(n-2b-1,t-2)- D_{2,b}(n-3b-2,t-3)+ D_{2,b}(n-2b-2,t-2)\\
      &\stackrel{(\diamond)}{\geq} D_{2,b}(n-2b-1,t-2)\\
      &\geq |\mathcal{U}^{\overline{x_1}}|,
    \end{align*}
    where $(\ast)$ follows by Corollary \ref{cor:del_con}, $(\star)$ follows by Lemma \ref{lem:del_rec}, and $(\diamond)$ follows by Lemma \ref{lem:del_inequality}.
    As a result, when $|\mathcal{U}^{x_1}|\leq N_{2,b}^-(n-1,t)$, by Corollary \ref{cor:del_int_rec}, we have $|\mathcal{U}^{\overline{x_1}}| > N_{2,b}^-(n-b-1,t-1)$, implying that $x_{b+1}= \overline{x_1}$.
\end{IEEEproof}

\begin{theorem}
  Let $n\geq b(t+1)-1$ with $t\geq 1$ and $b\geq 2$. Given $\boldsymbol{x}\in \Sigma_2^n$ and $\mathcal{U}=\{\boldsymbol{y}_1,\boldsymbol{y}_2,\ldots,\boldsymbol{y}_N\}\subseteq \mathcal{D}_{t,b}(\boldsymbol{x})$ with $|\mathcal{U}|=N\geq N_{2,b}^-(n,t)+1$, Algorithm \ref{alg:del} outputs the correct $\boldsymbol{x}$ and runs in $O_{t,b}(nN)$ time.
\end{theorem}

\begin{IEEEproof}
    Clearly, the correctness of the loop from Step 2 to Step 20 is supported by Lemma \ref{lem:del_con}. After this loop, we can determine $x_i$ for $i\notin \mathcal{T}$. 
    We then define the candidate set of sequences $\mathcal{V}= \{\boldsymbol{v}\in \Sigma_2^n: v_{i}= x_i \text{ for } i\notin \mathcal{T}\}$, which contains $\boldsymbol{x}$.
    According to Theorem \ref{thm:del_int}, there is exactly one sequence whose $b$-burst-deletion ball of radius $t$ can contain the set $\{\boldsymbol{y}_1,\boldsymbol{y}_2,\ldots,\boldsymbol{y}_N\}$.
    As a result, the algorithm outputs the correct $\boldsymbol{x}$.
    
    We now analyze the complexity of this algorithm.
    Firstly, in each iteration of the loop from Step 2 to Step 20, specifically in Step 4, 
    we need to calculate $|\mathcal{U}^{0}|$ and $|\mathcal{U}^{1}|$ from a row vector of length at most $N$. This process runs in $O_{q,t,b}(N)$ time. 
    Since the loop iterates at most $n$ times, the total complexity of this loop is $O_{q,t,b}(nN)$.
    Observe that $|\mathcal{T}|\leq t(b-1)$, which implies that $|\mathcal{V}|\leq 2^{t(b-1)}$.
    To complete the proof, it suffices to show that checking whether $\boldsymbol{y}_i$, for $i\in [1,N]$, belongs to $\mathcal{D}_{t,b}(\boldsymbol{v})$, for $\boldsymbol{v}\in \mathcal{V}$, has a complexity of $O_{t,b}(n)$.
    
    We compare $\boldsymbol{y}_i= y_{i,1} y_{i,2} \cdots y_{i,n-bt}$ and $\boldsymbol{v}$ symbol by symbol.
    Let $j$ be the largest index such that $\boldsymbol{y}_i$ and $\boldsymbol{v}$ share a common prefix of length $j$.
    Thus, we have $y_{i,1}\cdots y_{i,j}= v_{1}\cdots v_j$ and $y_{i,j+1}\neq v_{j+1}$.
    Let $f$ be the smallest index such that $y_{j+fb+1}= v_{j+1}$ if it exists, otherwise we set $f=\frac{n}{b}$.
    Observe that $f\geq 1$ and $\boldsymbol{y}_{i}\in \mathcal{D}_{t,b}(\boldsymbol{v})$ if and only if $\boldsymbol{y}_{i,[j+fb+2,n-bt]}\in \mathcal{D}_{t-f,b}(\boldsymbol{v}_{[j+2,n]})$.
    We then consider $\boldsymbol{y}_{i,[j+fb+2,n-bt]}$ and $\boldsymbol{v}_{[j+2,n]}$.
    After at most $t$ such operations, we can determine whether $\boldsymbol{y}_{i}\in \mathcal{D}_{t,b}(\boldsymbol{v})$.
    Clearly, this process runs in $O_{t,b}(n)$ times, thereby completing the proof.
\end{IEEEproof}

\section{Conclusion}
In this paper, we study the sequence reconstruction problem under channels that experience multiple bursts of insertions and multiple bursts of deletions, respectively.
We provide a complete solution for the insertion case and a partial solution for the deletion case. More specifically, we completely address the deletion case for the binary alphabet while discussing several aspects for the non-binary alphabet.
The following questions are intriguing and will be explored in future research:
\begin{itemize}
  \item Determining the exact value of $N_{q,b}^-(n,t)$ for $q> 2$, $t\geq 1$, and $b\geq 2$. We conjecture that the lower bound derived in Lemma \ref{lem:del_lb} is tight. In other words, we hypothesize that
      \begin{align*}
        N_{q,b}^-(n,t)= D_{q,b}(n,t) - D_{q,b}(n-b,t)+ D_{q,b}(n-(q+1)b,t-q).
      \end{align*}
  \item Investigating the model where each sequence is transmitted from some codebook $\mathcal{C}\subsetneq \Sigma_q^n$ instead of the full space $\Sigma_q^n$.
\end{itemize}

\end{document}